\documentclass[letterpaper]{IEEEtran}
\usepackage{amsfonts}
\usepackage[dvips]{graphicx}
\usepackage{epsfig,latexsym}
\usepackage{float}
\usepackage{indentfirst}
\usepackage{amsmath}
\usepackage{amssymb}
\usepackage{times}
\usepackage{subfigure}
\usepackage{stfloats,algorithm,algorithmic,bm}
\usepackage{psfrag}
\usepackage{cite}
\usepackage{subfigure}
\usepackage{xcolor}
\usepackage{amsfonts}
\usepackage[dvips]{graphicx}
\usepackage{epsfig,latexsym}
\usepackage{float}
\usepackage{indentfirst}
\usepackage{cases}
\usepackage{algorithmic,algorithm}

\usepackage{amsfonts,subfigure,multicol,color,verbatim,
graphicx,cite,epsfig,amssymb,amsmath}
\newtheorem{theorem}{Theorem}
\newtheorem{lemma}[theorem]{Lemma}

\newenvironment{proof}[1][Proof]{\begin{trivlist}
\item[\hskip \labelsep {\bfseries #1}]}{\end{trivlist}}
\newenvironment{definition}[1][Definition]{\begin{trivlist}
\item[\hskip \labelsep {\bfseries #1}]}{\end{trivlist}}

\begin{document}
\markboth{IEEE JOURNAL ON SELECTED AREAS IN COMMUNICATIONS, VOL. 33,
NO. X, Second Quarter, 2015} {Peng et al: Heterogeneous Cloud Radio
Access Networks\ldots}

\title{Contract-Based Interference Coordination in Heterogeneous Cloud Radio Access Networks
\author{Mugen~Peng,~\IEEEmembership{Senior~Member,~IEEE}, Xinqian Xie, Qiang Hu, Jie Zhang~\IEEEmembership{Member,~IEEE}, and H. Vincent Poor,~\IEEEmembership{Fellow,~IEEE}
\thanks{Manuscript received Jul. 21, 2013; revised Dec. 15, 2014; accepted
Feb. 19, 2015. The work of M. Peng, X. Xie, and Q. Hu was supported
in part by the National Natural Science Foundation of China under
Grant 61222103 and 61361166005, the National High Technology
Research and Development Program of China under Grant 2014AA01A701,
and the Beijing Natural Science Foundation under Grant 4131003. The
work of H. V. Poor was supported in part by the U.S. National
Science Foundation under Grant ECCS-1343210. Corresponding Author:
Mugen Peng, Tel.: +86-0-13911201899, Fax: +86-10-62283385.}
\thanks{ Mugen Peng (e-mail: pmg@bupt.edu.cn), Xinqian Xie (e-mail: xxmbupt@gmail.com), and Qiang Hu
(e-mail: hq760001570@gmail.com) are with the Key Laboratory of
Universal Wireless Communications, Ministry of Education, Beijing
University of Posts \text{\&} Telecommunications, Beijing, China.
Jie Zhang (e-mail: jie.zhang@sheffield.ac.uk) is with the Electronic
and Electrical Engineering Department at the University of
Sheffield, Sheffield, UK. H. Vincent Poor (e-mail:
poor@princeton.edu) is with the School of Engineering and Applied
Science, Princeton University, Princeton, NJ, USA.}}}

\maketitle

\begin{abstract}

Heterogeneous cloud radio access networks (H-CRANs) are potential
solutions to improve both spectral and energy efficiencies by
embedding cloud computing into heterogeneous networks (HetNets). The
interference among remote radio heads (RRHs) can be suppressed with
centralized cooperative processing in the base band unit (BBU) pool,
while the inter-tier interference between RRHs and macro base
stations (MBSs) is still challenging in H-CRANs. In this paper, to
mitigate this inter-tier interference, a contract-based interference
coordination framework is proposed, where three scheduling schemes
are involved, and the downlink transmission interval is divided into
three phases accordingly. The core idea of the proposed framework is
that the BBU pool covering all RRHs is selected as the principal
that would offer a contract to the MBS, and the MBS as the agent
decides whether to accept the contract or not according to an
individual rational constraint. An optimal contract design that
maximizes the rate-based utility is derived when perfect channel
state information (CSI) is acquired at both principal and agent.
Furthermore, contract optimization under the situation where only
the partial CSI can be obtained from practical channel estimation is
addressed as well. Monte Carlo simulations are provided to confirm
the analysis, and simulation results show that the proposed
framework can significantly increase the transmission data rates
over baselines, thus demonstrating the effectiveness of the proposed
contract-based solution.
\end{abstract}
\begin{keywords}
Heterogeneous cloud radio access networks, interference
coordination, contract-based game theory.
\end{keywords}

\section{Introduction}

The demand for high-speed and high-quality data applications is
expected to increase explosively in the next generation cellular
network, also called the fifth generation (5G), along with the rapid
growth of capital expenditure and energy consumption \cite{1}. Since
the traditional third generation (3G) and fourth generation (4G)
cellular networks, originally devised for enlarging the coverage
areas and optimized for the homogenous traffic, are reaching their
limits, heterogeneous networks (HetNets) have attracted intense
interest from both academia and industry \cite{2}. HetNets offer the
advantages of serving dense customer populations in hot spots with
low power nodes (LPNs) (e.g., pico base stations, femto base
stations, small cell base stations, etc.), while providing
ubiquitous coverage with macro base stations (MBSs), and reducing
the energy consumption \cite{3}. Unfortunately, high densities of
LPNs incur severe interference, which restricts performance gains
and commercial applications of HetNets \cite{4}. To suppress
inter-LPN and inter-tier interference, coordinated multi-point
(CoMP) transmission and reception has been presented as one of the
most promising techniques in 4G systems \cite{5}. Though performance
gains of CoMP are significant under ideal assumptions, they degrade
with increasing density of LPNs due to the non-ideal information
exchange and cooperation among LPNs. On the one hand, it has been
demonstrated that the performance of the downlink CoMP with
zero-forcing beamforming (ZFBF) depends heavily on time delay, and
thus the ZFBF design should be fairly conservative \cite{6}.
Further, it has been shown that CoMP ZFBF has no throughput gain
when the overhead channel delay is larger than 60\% of the channel
coherence time. On the other hand, it has been reported in \cite{7}
that the average spectral efficiency (SE) gain of uplink CoMP in
downtown Dresden field trials is only about 20\% with non-ideal
backhaul and distributed cooperation processing located on MBSs.
Hence, novel system architectures and advanced signal processing
techniques are needed to fully realize the potential gains of
HetNets.

Meanwhile, cloud computing technology has emerged as a promising
solution for providing good performance in terms of both SE and
energy efficiency (EE) across software defined wireless
communication networks \cite{8}. By leveraging cloud computing
technologies, the storage and computation originally provided in the
physical layer can be migrated into the ``cloud" to avoid redundant
resource consumption and to achieve the overall optimization of
resource allocation via centralized processing \cite{9}. As an
application of cloud computing to radio access networks, the cloud
radio access network (C-RAN) has been proposed to achieve
large-scale cooperative processing gains, though the constrained
fronthaul link between the remote radio head (RRH) and the baseband
unit (BBU) pool presents a performance bottleneck \cite{10}. Since
the C-RAN is mainly utilized to provide high data rates in hot
spots, real-time voice service and control signalling are not
efficiently supported. In order to avoid the significant signalling
overhead through decoupling the user plane and control plane in
RRHs, the traditional C-RAN must be enhanced and even evolved.

Motivated by the aforementioned challenges of both HetNets and
C-RANs, an advanced solution known as the heterogeneous cloud radio
access network (H-CRAN) is introduced in \cite{11}. The motivation
of H-CRANs is to avoid inter-LPN interference existing in HetNets
through connecting to a ``signal processing cloud'' with high-speed
optical fibers\cite{12}. As such, the baseband datapath processing
as well as the radio resource control for LPNs are moved to the BBU
pool so as to take advantage of cloud computing capabilities, which
can fully exploit the cooperation processing gains, lower operating
expenses, improve SE performance, and decrease energy consumption of
the wireless infrastructure. The simplified LPNs are converted into
RRHs, i.e., only the front radio frequency (RF) and simple symbol
processing functionalities are configured in RRHs, while the other
important baseband physical processing and procedures of the upper
layers are jointly executed in the BBU pool.

In H-CRANs, the control and user planes are decoupled, and the
delivery of control and broadcast signalling is shifted from RRHs to
MBSs, which alleviates the capacity and time delay constraints on
the fronthaul. Therefore, RRHs are mainly used to provide high bit
rates with high EE performances, while the MBS is deployed to
guarantee seamless coverage and deliver the overall control
signalling. With the help of MBSs, unnecessary handover and
re-association can be avoided. The adaptive signaling/control
mechanism between connection-oriented and connectionless modes is
supported in H-CRANs, which can achieve significant overhead savings
in the radio connection/release by moving away from a pure
connection-oriented mechanism.

As shown in Fig. \ref{fig:system model}, the BBU pool is interfaced
to MBSs for coordinating the cross-tier interference between RRHs
and MBSs. The data and control interfaces between the BBU pool and
the MBS are S1 and X2, respectively, which are inherited from
definitions of the $3^{rd}$ generation partnership project (3GPP)
standards. With the help of centralized large-scale cooperative
processing in the BBU pool, the intra-tier interference can be fully
suppressed, while the inter-tier interference between RRHs and MBSs
still present challenges to improving SE and EE performance.
Significant attention has been paid to inter-tier interference
collaboration \cite{13, 13_1} and radio resource cooperative
management techniques \cite{14} to alleviate inter-tier
interference. As these works are focused primarily on HetNets, new
models and methods are needed for H-CRANs.

\begin{figure}[!t!h]
\center
\includegraphics[width=0.45\textwidth]{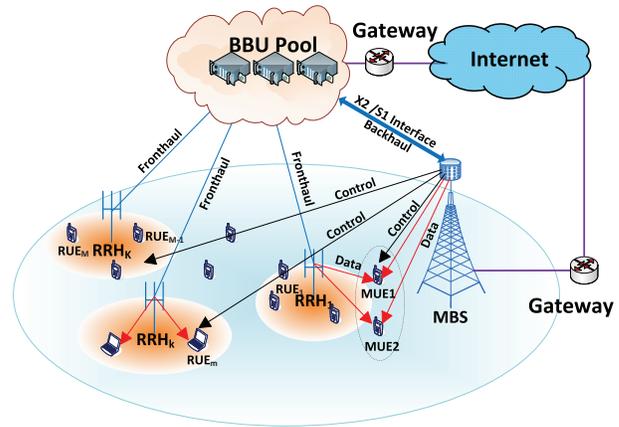}
\caption{System model of the H-CRAN with one MBS and $K$ RRHs.}
\label{fig:system model}
\end{figure}

Intuitively, game theory is a useful tool for examining the
inter-tier interference problem, and the published literature for
HetNets has presented using this approach. For example, a
Stackelberg game model has been used in \cite{15} to design a
distributed radio resource algorithm, and its adaption to
hierarchical problems such as hierarchical power control has been
considered in \cite{16}. More specifically, a Stackelberg game with
a single leader and multiple followers is formulated to study joint
objective optimization of the macrocell and femtocells subject to a
maximum interference power constraint at the MBS in \cite{17}. In
\cite{18}, a hierarchical game framework with multiple leaders and
multiple followers is adopted to investigate uplink power allocation
to mitigate inter-tier interference in two-tier femtocell networks.
In order, for the Stackelberg game, to reach an equilibrium, players
involved in the game have to interact periodically, which leads to
high latency and signalling overhead. Furthermore, the incentive in
the Stackelberg game model is ineffective if the information is
asymmetric between the leader and follower. To overcome these
issues, an advanced game model with proper incentive mechanisms and
regarding the limited overhead is preferred.

To overcome the aforementioned problems in the Stackelberg model,
contract-based game theory has recently been applied to aspects of
spectrum sharing and relay selection. In \cite{19}, to solve the
cooperative spectrum sharing problem under incomplete channel state
information (CSI) in cognitive radio networks, a resource-exchange
scheme is proposed and the corresponding optimal contract model is
designed. The problem of spectrum trading with a single primary user
selling its idle spectrum to multiple secondary users is
investigated in \cite{20}, where a money-exchange based contract is
offered by the primary user containing a set of different items each
intended for a specific consumer type. In \cite{21}, under an
asymmetric information scenario where the source is not well
informed about the CSI of potential relays, a contract based game
theory model is designed to help the source select relays for
optimizing the throughput. Motivated by these existing works, since
the cooperative processing capabilities and acquired CSI are
significantly different between the BBU pool and the MBS in H-CRANs,
contract-based game theory is a promising technique for modeling the
inter-tier interference coordination mechanism under information
asymmetry.

Note that ideal inter-tier interference coordination depends
critically on having accurate CSI for all transmitters and receivers
\cite{22}; however, this assumption is non-trivial especially in
view of the overwhelming signalling overhead and the inevitable
existence of CSI distortion in practical systems \cite{23}. In
previous studies, CSI errors are typically modeled as zero-mean
complex Gaussian random variables without taking the practical
training process and channel estimation methods into account
\cite{24}. In order to present a realistic scenario, not only the
impact of imperfect CSI but also the corresponding training based
channel estimation technique should be considered when examining
contract-based interference coordination in H-CRANs.

In this paper, we propose a contract-based interference coordination
framework to mitigate the inter-tier interference between RRHs and
MBSs in H-CRANs. Based on the proposed framework, different
scheduling schemes can be utilized jointly to achieve performance
gains. Specifically, three scheduling schemes are embedded into the
proposed framework in this paper, and the downlink transmission is
divided into three phases accordingly. At the first phase, RRHs help
MBSs serve macro-cell user equipments (MUEs), which is termed the
\emph{RRH-alone with UEs-all} scheme. At the second phase, MBSs let
RRHs use all radio resources to serve the associated UEs (denoted by
RUEs), which is called the \emph{RRH-alone with RUEs-only} scheme.
Then at the third phase, RRHs and MBSs serve their own UEs
separately, namely the \emph{RRH-MBS with UEs-separated} scheme.
Note that these three schemes are examples merely used to show the
principles of the proposed contract-based inter-tier interference
coordination framework, and other inter-tier interference
coordination schemes can be embedded into this framework, or replace
these three schemes. Through the contract-based interference
coordination framework, RRHs and MBSs establish a mutually
beneficial relationship, which can be effectively modeled and
optimized by using contract-based game theory. The main
contributions of this paper are summarized as follows:

\begin{itemize}
\item To coordinate the inter-LPN interference in HetNets, and to compensate for the substantial fronthaul overhead and processing in conventional C-RAN,
the H-CRAN is considered in this paper, which incorporates
advantages of both the HetNet and C-RAN. The H-CRAN decouples the
control and user planes, and thus RRHs are mainly used to provide
high data rates via cloud computing to suppress the intra-tier
interference, while the MBS is deployed to guarantee seamless
coverage and deliver all control and broadcast signals.
\item Since the cooperative processing capability and acquired CSI between the BBU pool and MBS are asymmetric, a contract-based interference coordination framework is proposed to coordinate the inter-tier interference between RRHs and MBSs, where the BBU pool and the MBS are
modeled as the principal and agent, respectively. Based on the
proposed coordination framework, three scheduling schemes, namely
\emph{RRH-alone with UEs-all}, \emph{RRH-alone with RUEs-only}, and
\emph{RRH-MBS with UEs-separated}, are adaptively utilized.
Accordingly, the downlink transmission in each interval is divided
into three phases.
\item With complete CSI, i.e., the BBU pool can acquire the perfect global CSI, sufficient and necessary
conditions for a feasible contract are presented. An contract design
that maximizes a rate-based utility is derived to achieve the
optimized transmission duration for these three phases and to obtain
the optimized received power allocation for RUEs and MUEs.
Theoretical analysis indicates that the optimal rate-based utility
is achieved when the time duration of the third phase is zero, i.e.,
the \emph{RRH-MBS with UEs-separated} scheme is not triggered due to
the large inter-tier interference even though the fairness power
control algorithm is applied.
\item Considering the imperfect CSI acquired in piratical H-CRANs, training and channel estimation schemes are presented to obtain partial CSI.
Based on the estimated partial CSI, contract design and optimization
for the presented framework are addressed. Specifically, sufficient
and necessary conditions for a feasible contract under partial CSI
are presented, which are individually rational and incentive
compatible. In addition, the optimization problem to determine the
contract that maximizes the rate-based utility is formulated.
Moreover, an optimal contract design is derived when the number of
RRHs is large.
\end{itemize}

The rest of this paper is organized as follows: Section II describes
the system model and the proposed interference coordination
framework. In Section III, an optimal contract design with ideal and
complete CSI is presented. Then in Section IV, contract design with
practical channel estimation is considered. Numerical results are
shown in Section V, followed by conclusions in Section VI.

\emph{Notation}: The transpose, inverse, and Hermitian of matrices
are denoted by $\left(\cdot\right)^{T}$, $\left(\cdot\right)^{-1}$,
and $\left(\cdot\right)^{H}$, respectively. $\|\cdot\|$ denotes the
two-norm of vectors.
$\mathbf{diag}\left(e_{1},e_{2},\ldots,e_{N}\right)$ is the $N\times
N$ diagonal matrix. $\mathbf{I}_{K}$ represents the $K\times K$
dimensional unitary diagonal matrix. $\mathcal{E}\{\cdot\}$ denotes
the expectation of random variables and $\dot{\left(\cdot\right)}$
denotes the first order derivative. For convenience, the
abbreviations are listed in Table \ref{table1}.

\begin{table}\label{table1}
\center \caption{Summary of Abbreviations}
\begin{tabular}{l l}\hline
AWGN & additive white Gaussian noise\\
BBU & base band unit\\
cdf & cumulative density function\\
CICF & contract-based interference coordination framework\\
CoMP & coordinated multi-point\\
C-RAN & cloud radio access network\\
CSI & channel state information\\
EE & energy efficiency\\
FRPC & frequency reuse with power control\\
H-CRAN & heterogeneous cloud radio access network\\
HetNet    & heterogeneous network\\
IC & incentive compatible\\
IR & individual rational\\
KKT & Karush-Kuhn-Tucker\\
LPN & low power node\\
LSE & least-squares estimate\\
MBS & macro base station\\
MSE & mean-square error\\
MUE & macro-cell user equipment\\
OFDMA & orthogonal frequency division multiple access\\
pdf  & probability density function\\
RB & resource block\\
RF & radio frequency\\
RRH & remote radio head\\
RUE & remote radio head served user equipment\\
SE & spectral efficiency\\
SNR & signal-to-noise ratio\\
TDD & time division duplex\\
TDIC & time domain interference cancelation\\
TTI & time transmission interval\\
UE & user equipment\\
ZFBF & zero-forcing beamforming\\\hline
\end{tabular}
\end{table}

\section{System Model and Interference Coordinated Framework}

The discussed H-CRAN consists of one BBU pool, $K$ RRHs and one MBS,
in which each RRH connects with the BBU pool through an ideal
optical fiber, and the MBS connects with the BBU pool ideally with
no constraints. It is assumed that all $K$ RRHs cooperate, and they
share the same radio resources with the MBS. The radio resources are
allocated to different MUEs using orthogonal frequency division
multiple access (OFDMA). It is assumed that each node has a single
antenna, and thus for each resource block (RB) in the OFDMA based
H-CRAN, only one MUE will be served by the MBS, and at most $K$ RUEs
will be associated with the $K$ RRHs. Without loss of generality,
$M$ ($M<K$) RUEs and one MUE can share the same RB in OFDMA based
H-CRAN systems. When $M\!\geq\!K$, $K$ RUEs are selected to be
auto-scheduled \cite{26}. We note that a number of synchronization
solutions for distributed antenna and virtual MIMO systems, e.g.,
the protocol in \cite{28} and post-facto synchronization \cite{29},
can be effectively applied to accomplish the synchronization among
RRHs. Thus we assume that all RRHs and the MBS can be perfectly
synchronized.

Downlink transmission is considered, in which the BBU pool sends
individual data flows to RUEs via RRHs, and the MBS transmits the
data flows to MUEs. The radio channel matrix between RRHs and RUEs
is denoted by $\mathbf{G}$ whose $\left(m,k\right)$-th entry
$g_{mk}$ represents the channel coefficient between the $k$-th RRH
and the $m$-th RUE. The channel coefficient between the MBS and MUE
is denoted by $g_{B}$. Let $f_{kM}$ denote the radio channel
coefficient between the $k$-th RRH and MUE, and $f_{Bm}$ represent
the channel coefficient between the $m$-th RUE and MBS. It is
assumed that all channels are quasi-static flat fading so that they
remain constant within one transmission block but may vary from one
block to another. We further define $s_{m}$ as the symbol
transmitted from the BBU pool to the $m$-th RUE, and $s_{M\!+\!1}$
as the symbol from the MBS to MUE. It is assumed that distributed
precoding is used at all $K$ RRHs, and the $K\!\times\!M$ precoding
matrix is denoted as $\mathbf{A}$. The observations at all $M$ RUEs
can be expressed in vector form as
\begin{align}\label{EQ1}
{{\bf{y}}_R} = {\bf{GAs}}\, + \underbrace {{{\bf{f}}_B}{s_{M + 1}}}_{{\rm{interference}}} + \,{{\bf{n}}_C},
\end{align}
where $\mathbf{s}\!=\!\left[s_{1},s_{2},\ldots,s_{M}\right]^{T}$,
$\mathbf{f}_{B}\!=\!\left[f_{B1},f_{B2},\ldots,f_{BM}\right]^{T}$,
and $\mathbf{n}_{C}$ is an $M\!\times \!1$ dimensional additive
white Gaussian noise (AWGN) vector with each entry having the zero
mean and variance $\sigma_{n}^{2}$. Moreover,
$\mathbf{G}\!=\!\mathbf{D}^{\frac{1}{2}}\mathbf{H}$ models the
independent fast fading and large-scale fading, where $\mathbf{H}$
is an $M\times K$ matrix of fast fading coefficients whose entries
are independent and identically distributed (i.i.d.) zero-mean
complex Gaussian with unity variance, and $\mathbf{D}^{\frac{1}{2}}$
is an $M\times M$ diagonal matrix with
$[\mathbf{D}]_{mm}\!=\!\upsilon_{m}$ known at the BBU pool. The
variance of the $m$-th element in $\mathbf{f}_{B}$ is denoted by
$\upsilon_{Bm}$. The signal received at the MUE is given by
\begin{align}\label{EQ2}
y_B\!=\!g_{B}s_{M\!+\!1}\!
+\!\underbrace{\mathbf{f}_{M}\mathbf{A}\mathbf{s}}_{\mathrm{interference}}\!+\!n_{B},
\end{align}
where $\mathbf{f}_{M}\!=\!\left[f_{1M},f_{2M},\ldots,f_{KM}\right]$
with variances of elements denoted by $\upsilon_{M\!+\!1}$, and
$n_{B}$ is zero-mean Gaussian noise with variance $\sigma_{n}^{2}$.
Eqs. \eqref{EQ1} and \eqref{EQ2} suggest that the received signals
at the RUEs and MUE experience the inter-tier interference, thus
degrading the transmission performance.

To coordinate this inter-tier interference, three scheduling schemes
are considered as examples, namely \emph{RRH-alone with UEs-all},
\emph{RRH-alone with RUEs-only}, and \emph{RRH-MBS with
UEs-separated}. For the \emph{RRH-alone with UEs-all} scheme, both
RUEs and MUEs are served by RRHs, and the MBS keeps silent. For the
\emph{RRH-alone with RUEs-only} scheme, to avoid the interference
and interruption to RUEs, both MBS and MUE keep silent, and only
RRHs and their corresponding serving RUEs operate. While for the
\emph{RRH-MBS with UEs-separated} scheme, both RRHs and the MBS
operate simultaneously, and the interference between RUEs and the
MUE is mitigated by using power control. To adaptively adopt these
three schemes, a contract-based cooperative transmission framework
is presented to coordinate these three scheduling schemes.
Accordingly, as shown in Fig. \ref{fig:framestructure}, there are
three phases for these three scheduling schemes in each time
transmission interval (TTI) with a fixed time length $T_{0}$.

\begin{figure}[!t!h]
\center
\includegraphics[width=0.5\textwidth]{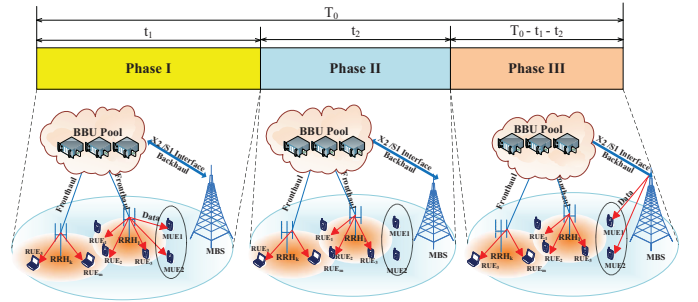}
\caption{Frame structure and scheduling algorithm}
\label{fig:framestructure}
\end{figure}

\begin{itemize}
\item \textbf{Phase I} : During the time period of $\left[0, t_{1}\right]$ in each TTI, the \emph{RRH-alone with UEs-all} scheme operates, where RRHs help the MBS serve the MUE, and the MBS benefits RRHs. Particularly, all $K$ RRHs serve $M$ RUEs and the single MUE in the typical RB, while the MBS remains idle. Since $K\geq M\!+\!1$ is satisfied, the interference between RUEs and the MUE can be suppressed by using distributed precoding in the BBU pool. Note that RRHs consume additional power for serving the MUE, which can be regarded as the payoff of the RRHs.
\item \textbf{Phase II} : During the time period of $\left[t_{1}, t_{1}\!+\!t_{2}\right]$ in each TTI, the \emph{RRH-alone with RUEs-only} scheme operates, where all $K$ RRHs only serve the accessed $M$ RUEs, while the MBS remains idle and the corresponding MUE is not served. Since RRHs help the MBS to serve the MUE in Phase I, the MBS awards RRHs by staying idle at this phase; thus there is no inter-tier interference from the MBS to RUEs.
\item \textbf{Phase III} : During the remaining time period of $\left[t_{1}\!+\!t_{2}, T_{0}\right]$ in each TTI, the \emph{RRH-MBS with UEs-separated} scheme operates, where RRHs and the MBS serve RUEs and the MUE,
    respectively. To mitigate the inter-tier interference, the fairness power control is utilized herein.
\end{itemize}

The H-CRAN architecture can support efficiently the contract-based
interference coordination framework, in which the BBU pool manages
all connected RRHs and accessed RUEs. The inter-tier interference
coordination happens between the BBU pool and the MBS, which avoids
the substantial overhead and capacity requirements between LPNs and
MBSs of HetNets. Note that the advanced inter-tier interference
collaboration scheme \cite{13} and radio resource cooperative
management techniques \cite{14} can be used to replace the
\emph{RRH-MBS with UEs-separated} scheme at Phase III. To highlight
characteristics of the proposed contract-based cooperative
transmission framework and exploit performance gains resulting from
contract-based game theory in H-CRANs, scheduling without
coordinating the severe inter-tier interference is considered in the
Phase III.

\emph{Remark}: The BBU pool and the MBS have conflicting objectives
in the aforementioned three schemes. RRHs want to have a longer
$t_{2}$ because all RUEs can be served without inter-tier
interference during Phase II. Alternatively, the MBS wants a longer
$t_{1}$ because the MUE can be served by RRHs without its own power
consumption. If both RRHs and the MBS can benefit from this
presented transmission framework, they would cooperate with each
other. Otherwise, it would lead to $t_{1}\!=\!t_{2}\!=0$, and ${t_3}
= {T_0}$ which devolves to the traditional transmission scheme
enduring severe inter-tier interference. The core idea of the
proposed framework is that two participants under information
asymmetry, i.e., the BBU pool and MBS, establish a mutually
beneficial relationship, which can be effectively modeled and
optimized by using contract-based game theory.

\section{Optimal Contract Design under Complete CSI}

It is preferred that BBU pool managing on all RRHs and RUEs act as
the principal, due to the powerful capabilities of centralized cloud
computing and large-scale cooperative processing. In the
contract-based game model, the BBU pool as the principal designs a
contract and offers it to the agent. The MBS as the agent decides to
accept or reject the offered contract. Since the BBU pool and MBS
are connected with an X2/S1 interface in H-CRANs, it is possible for
them to exchange the corresponding individual CSI timely and
completely. Thus, we can assume that the BBU pool can acquire the
global and perfect CSI in this section.

\subsection{Rate-based Utility Definition}

A rated-based utility function for the contract is considered in
this paper to represent the SE performance of H-CRANs. Though more
advanced precoding techniques can be utilized to provide better
interference coordination performance with higher implementing
complexities than ZF does, the ZF precoding technique is adopted in
this paper to suppress the inter-RRH interference in the BBU pool
because this common linear precoding can achieve performance close
to the sum capacity bound with a relatively low complexity when the
number of RUEs is sufficiently large in interference-limited H-CRAN
systems \cite{14_Lau}.

\begin{itemize}
\item \textbf{Phase I}: Denote the received symbol power of $s_{m}$ and $s_{M\!+\!1}$ by $P_{C1}$ and $P_{M1}$, respectively. The sum transmission data rate of all RUEs, denoted by $R_{C1}$, and the data rate of the MUE, denoted by $R_{M1}$, can be expressed as
    \begin{align}
    R_{C1}\!&=\!M\log\left(1\!+\!\frac{P_{C1}}{\sigma_{n}^{2}}\right),\\
     R_{M1}\!&=\!\log\left(1\!+\!\frac{P_{M1}}{\sigma_{n}^{2}}\right),
    \end{align}
    respectively. Defining $\mathbf{F}\!=\![\mathbf{G}^{T},\mathbf{f}_{M}^{T}]^{T}$, the expected total transmit power of the RRHs is constrained by
    \begin{align}
    \mathcal{E}\left\{\textup{tr}\left(\mathbf{F}\mathbf{F}^{H}\right)^{-1}\mathbf{\Lambda}
    \right\}\!\leq\!P_{\mathrm{max}},
    \end{align}
    where $P_{\mathrm{max}}$ is the pre-defined allowable maximum transmit power, and $\mathbf{\Lambda}$ is an $(M\!+\!1)\times(M\!+\!1)$ diagonal matrix with the $(M\!+\!1)$-th
    diagonal element given by $P_{M1}$ and other diagonal elements given by $P_{C1}$. The transmit power constraint can be re-written as
    \begin{align}\label{power}
    \textup{tr}\left(\mathbf{\tilde{D}}^{-\frac{1}{2}}\mathcal{E}\left\{(\mathbf{\tilde{H}}
    \mathbf{\tilde{H}}^{H})^{-1}\right\}\mathbf{\tilde{D}}^{-\frac{1}{2}}\mathbf{\Lambda}\right)\!\leq\!P_{\mathrm{max}},
    \end{align}
    where $\mathbf{\tilde{D}}\!=\!\mathbf{diag}\left(\upsilon_{m},\ldots,\upsilon_{M+1}
    \right)$ is an $(M\!+\!1)\times(M\!+\!1)$ diagonal matrix with $\mathbf{\tilde{H}}\!=\![\mathbf{H}^{T},\mathbf{h}_{f}^{T}]^{T}$,
    where $\mathbf{h}_{f}$ is a $1\times K$ Gaussian vector with zero means and unity variances. As $\tilde{\mathbf{H}}\tilde{\mathbf{H}}^{H}\sim\mathcal{W}(M\!+\!1,\mathbf{I}_{M\!+\!1})$ is an $(M\!+\!1)\times(M\!+\!1)$ central complex Wishart matrix with $K(K\!>\!M\!+\!1)$ degrees of freedom, \eqref{power} can be
    written as
    \begin{align}
    \sum_{m=1}^{M}\frac{P_{C1}}{(K\!-\!M\!-\!1)\upsilon_{m}}\!+\!\frac{P_{M1}}{(K\!-\!M\!
    -\!1)\upsilon_{M\!+\!1}}\!\leq\!P_{\mathrm{max}}.
    \end{align}
\item \textbf{Phase II}: Denote the received symbol power of $s_{m}$ by $P_{C2}$, and the sum rate of all RUEs can be expressed
as
    \begin{align}
    R_{C2}\!=\!M\log\left(1\!+\!\frac{P_{C2}}{\sigma_{n}^{2}}\right),
    \end{align}
    where the long-term average transmit power is similarly constrained by
    \begin{align}
    \mathcal{E}\bigg\{P_{C2}\mathrm{tr}\left(\mathbf{G}\mathbf{G}^{H}\right)^{-1}\bigg\}=\sum_{m\!=\!1}^{M}\frac{P_{C2}}{\upsilon_{m}\left(K\!-\!M\right)}\!\leq\!P_{\mathrm{max}}.
    \end{align}

    Since there is no inter-tier interference from the MBS to the desired RUEs due to the silent MUE, $P_{C2}$ would be always set at the maximum value $P_{C2}\!=\!\frac{\left(K\!-\!M\right)
    P_{\mathrm{max}}}{\varepsilon_{1}}$ to optimize the sum rate, where $\varepsilon
    _{1}\!=\!\sum_{m=1}^{M}\frac{1}{\upsilon_{m}}$.

\item \textbf{Phase III}: Denote the transmit power of the MBS and the RRH
by $P_{B}$ and $P_{C3}$, respectively. Under the fairness power
control algorithm\cite{29}, $P_{C3}$ and $P_{B}$ are derived from
    \begin{align}
    \{P_{C3},{P_B}\}\!&=\!\arg \mathop {\max }\limits_{P_{C3},P_B} \min \left\{ {{R_{C3}},{R_{M3}}} \right\},\\
    s.t.\quad &0\!\leq\!P_{C3}\!\leq\!P_{\rm{max}},\thinspace 0\!\leq\!P_{B}\!\leq\!P_{\rm{max}}.
    \end{align}

    The sum rate of all RUEs, denoted by $R_{C3}$, and the data rate of the MUE, denoted by $R_{M3}$, can be expressed as
    \begin{align}
    {R_{C3}} = \sum\limits_{m = 1}^M {\log } \left( {1 + \frac{{{P_{C3}}}}{{|{f_{Bm}}{|^2}{P_B} + \sigma _n^2}}} \right),\\
    {R_{M3}} = \log \left( {1 + \frac{{|{g_B}{|^2}{P_B}}}{{{P_{C3}}{{\bf{f}}_M}{{\bf{G}}^H}{{\left( {{\bf{G}}{{\bf{G}}^H}} \right)}^{ - 2}}{\bf{Gf}}_M^H + \sigma _n^2}}} \right),
    \end{align}
\end{itemize}
respectively.

The rate-based utility of the BBU pool is defined as the sum rate in
each TTI obtained by adopting the proposed framework compared with
the traditional scheme in Phase III, which is given by
\begin{align}\label{EQ6}
{U_C} = {t_1}{R_{C1}} + \left( {{T_0} - {t_1} - {t_3}}
\right){R_{C2}} - \left( {{T_0} - {t_3}} \right){R_{C3}}.
\end{align}

Further, the data rate-based utility of the MBS during each TTI can
be expressed as
\begin{align}
{U_M} = {t_1}{R_{M1}} + {t_3}{R_{M3}}.
\end{align}

\subsection{Contract Design under Perfect CSI}

To make the contract feasible under complete CSI, the principal
should ensure that the agent obtains at least the reservation
utility value that it would get by accepting this contract. Such a
constraint is known as being individual rational (IR), which has the
following definition.

\begin{definition}
(\emph{Individual Rational}) A contract with the item $\left(
{{t_1},{t_3},{P_{M1}},{P_{C1}}} \right)$ is individual rational if
the rate-based utility that the agent obtains is larger than the
reservation utility ${u}$, i.e.,
\begin{align}
{t_1}{R_{M1}} + {t_3}{R_{M3}} \ge u = {T_0}{R_{M3}}.
\end{align}
\end{definition}

Since the principal desires to obtain a maximal profit from the
contract, the optimal contract $\left(
{t_1^*,t_3^*,P_{M1}^*,P_{C1}^*} \right)$ can be obtained from
\begin{equation}
\begin{split}
\left( {t_1^*,t_3^*,P_{M1}^*,P_{C1}^*} \right) & = \arg \mathop
{\max
}\limits_{{t_1},{t_3},{P_{M1}},{P_{C1}}} \bigg\{t_1 R_{C1} + \\
& \left( {{T_0} - {t_1} - {t_3}} \right){R_{C2}} - \left( {{T_0} - {t_3}} \right)R_{C3}\bigg\},\label{eq1}\\
s.t.\quad & {t_1}{R_{M1}} + {t_3}{R_{M3}} \geq {T_0}{R_{M3}},\\
& {t_3} \ge 0,{t_1} \ge 0,{T_0} - {t_1} - {t_3} \ge 0,\\
&
\frac{P_{C1}}{\left(K\!-\!M\!-\!1\right)\varepsilon_{1}}\!+\!\frac{P_{M1}}{\left(K\!-\!M\!-\!1\right)\upsilon_{M\!+\!1}}
\!\le\!P_{\rm{max}}.
\end{split}
\end{equation}

\begin{lemma}
Under complete CSI, the contract $\left(
{t_1^*,t_3^*,P_{M1}^*,P_{C1}^*} \right)$ that maximizes the
rate-based utility of the BBU pool satisfies the following
conditions:
\begin{align}
&t_1^*R_{M1}^* - \left( {{T_0} - t_3^*} \right)R_{M3}^* = 0,\\
&t_3^* = 0,\\
&\frac{P^{*}_{C1}}{\left(K\!-\!M\!-\!1\right)\varepsilon_{1}}\!+\!\frac{P^{*}_{M1}}
{\left(K\!-\!M\!-\!1\right)\upsilon_{M\!+\!1}}\!=\!P_{\mathrm{max}}.
\end{align}
\end{lemma}
\begin{proof}
See Appendix A.
\end{proof}

\textit{Remark}: Lemma 1 indicates that the MBS will obtain zero
utility gain from the contract, which is not unexpected since the
BBU pool will maximize its own utility with the minimal payoff.
Moreover, the time period of Phase III would be zero when the
contract is accepted by the MBS. Otherwise, the contract is
rejected.

Once the optimal contract is accepted, the objective function
depends only on $P_{M1}$, and the computational complexity for
obtaining $P^{*}_{M1}$ is acceptable with numerical methods, which
is provided in Algorithm 1. Once $P^{*}_{M1}$ is obtained,
$t_{1}^{*}$ and $P^{*}_{C1}$ can be directly determined.

\begin{table}
\center \caption{Algorithm 1: one-dimensional search for
$P_{M1}^{*}$ under perfect CSI}
\begin{tabular}{p{3.4in}@{}}
\hline
\textbf{Input:} $P_{M1}^{(1)}$, $P_{M1}^{(2)}$, where $P_{M1}^{(1)},
P_{M2}^{(2)}\in[0, (K\!-\!M\!-\!1)\upsilon_{M+1}P_{\rm{max}}]$\\
\quad satisfy $P_{M1}^{(2)}>P_{M1}^{(1)}$,
$\varsigma_{1}=\frac{\partial {U}_{C}} {\partial
P_{M1}}\big|_{P_{M1}=P_{M1}^{(1)}}<0$,\\ \quad and
$\varsigma_{2}=\frac{\partial {U}_{C}} {\partial P_{M1}}\big|_{P_{M1}=P_{M1}^{(2)}}>0$.\\
\textbf{Initialize:} $U_{C}^{(1)}$, $U_{C}^{(2)}$ by substituting $P_{M1}^{(1)}$ and $P_{M1}^{(2)}$ into (14), respectively.\\
\quad $k=1$, error $\epsilon\!>\!0$.\\
\thinspace \textbf{Repeat until convergence:} \\
\enspace Calculate:\\
\enspace
$s=\frac{3[U_{C}^{(2)}-U_{C}^{(1)}]}{P_{M1}^{(2)}-P_{M1}^{(1)}}$,\thinspace$z=s-\varsigma_{1}-\varsigma_{2}$,
\thinspace$w=\sqrt{z^{2}-\varsigma_{1}\varsigma_{2}}$,\thinspace \\
$\bar{P}_{M1}=P_{M1}^{(1)}+(P_{M1}^{(2)}-P_{M1}^{(1)})\left(1-\frac{\varsigma_{2}+w+z}
{\varsigma_{2}-\varsigma_{1}+2w}\right)$.\\
\enspace \textbf{If} $|P_{M1}^{(2)}-P_{M1}^{(1)}|\leq\epsilon$\\
\quad Termination criterion is satisfied.\\
\enspace \textbf{Else}\\
\quad Calculate $\bar{U}_{C}$ and $\varsigma=
\frac{\partial U_{C}}{\partial P_{M1}}\big|_{P_{M1}=\bar{P}_{M1}}$ by substituting $\bar{P}_{M1}$ into (14).\\
\enspace\textbf{If} $\varsigma=0$\\
\quad Termination criterion is satisfied.\\
\enspace\textbf{Else}\\
\quad\textbf{If} $\varsigma<0$\\
\qquad$P_{M1}^{(1)}=\bar{P}_{M1}$, $U_{C}^{(1)}=\bar{U}_{C}$ and $\varsigma_{1}=\varsigma$,\\
\quad\textbf{If} $\varsigma>0$\\
\qquad$P_{M1}^{(2)}=\bar{P}_{M1}$, $U_{C}^{(2)}=\bar{U}_{C}$ and $\varsigma_{2}=\varsigma$,\\
\enspace$k\!=\!k\!+\!1$.\\\hline
\end{tabular}
\end{table}

\emph{Remark}: Algorithm 1 is based on the spline interpolation
method, whose computational complexity is mainly determined by the
number of iterations. In this case, the computational complexity can
be expressed as $O\left(k\right)$, where $k$ denotes the number of
iterations. Note that $k$ depends on the error convergence parameter
$\epsilon$ and the initial configuration of $P_{M1}$ critically;
thus it is important to choose proper initial $P_{M1}$ to promote
the convergence of the algorithm.

Of course, complete and perfect CSI cannot usually be achieved in
practice. In practical H-CRANs, CSI must be estimated and
quantified, and is not error-free or delay-free. In the following
section, we consider the impact of the imperfect CSI on the
contract-based framework, where the principal can only acquire
partial and non-ideal CSI via channel estimation.

\section{Contract Design under Practical Channel Estimation}

In order to enable the BBU pool to acquire the downlink
instantaneous CSI, an uplink training design is considered, which is
based on the channel reciprocity in the time division duplex (TDD)
mode. It is assumed that the fronthaul and backhaul are
non-constrained and ideal, while the radio access links are
non-ideal.

\subsection{Uplink Training Design and Channel Estimation}

Let $\boldsymbol{\psi}_{m}$ denote the training sequence transmitted
from the $m$-th RUE to the corresponding RRH, and
$\boldsymbol{\psi}_{M\!+\!1}$ represents the training sequence from
the MUE to the MBS. All $\boldsymbol{\psi}_{m}$'s have the $N$($>M$)
symbol length. $\|\boldsymbol{\psi}_{m}\|^{2}\!=\!NP_{s}$ and
$\|\boldsymbol{\psi}_{M\!+\!1}\|^{2}\!=\!NP_{b}$, where $P_{s}$ and
$P_{b}$ are the training power of the RUEs and MUE, respectively.
Moreover,
$\mathbf{\Psi}\!=\!\left[\boldsymbol{\psi}_{1},\boldsymbol{\psi}_{2},\ldots,\boldsymbol{\psi}_{M}\right]$
satisfies $\mathbf{\Psi}^{H}\mathbf{\Psi}\!=\!NP_{s}\mathbf{I}_{M}$
and
$\mathbf{\Psi}^{H}\boldsymbol{\psi}_{M\!+\!1}\!=\!\mathbf{0}_{M}$.
During the training stage in the uplink, all RUEs and the MUE
transmit their own training sequences simultaneously, and the
observations at the $k$-th RRH can be expressed as
\begin{align}
\mathbf{x}_{k}\!=\!\sum_{m\!=\!1}^{M}g_{mk}\boldsymbol{\psi}_{m}\!+\!f_{kM}\boldsymbol{\psi}_{M\!+\!1}\!+\!\mathbf{w}_{k},
\end{align}
where $\mathbf{w}_{k}$ is an $N\!\times\!1$ dimensional AWGN vector
with covariance $\sigma_{n}^{2}\mathbf{I}_{N}$. The least-squares
estimate (LSE) of the channel coefficient between the $k$-th RRH and
$m$-th RUE $g_{mk}$, denoted by $\hat{g}_{mk}$, can be expressed as
\begin{align}
\hat{g}_{mk}\!=\!\frac{\boldsymbol{\psi}^{H}_{m}}{\|\boldsymbol{\psi}_{m}\|^{2}}\mathbf{x}_{k},
\end{align}
and the LSE of the channel coefficient between $k$-th RRH and MUE
$f_{kM}$ is
\begin{align}
\hat{f}_{kM}\!=\!\frac{\boldsymbol{\psi}^{H}_{M\!+\!1}}{\|\boldsymbol{\psi}_{M\!+\!1}\|^{2}}\mathbf{x}_{k}.
\end{align}

Similarly, the MBS observes
\begin{align}
\mathbf{x}_{B}\!=\!\boldsymbol{\psi}_{M\!+\!1}g_{B}\!+\sum_{m\!=\!1}^{M}\boldsymbol{\psi}_{m}f_{Bm}\!+\!\mathbf{w}_{B},
\end{align}
where $\mathbf{w}_{B}$ is an $N\!\times\!1$ dimensional AWGN vector
with covariance $\sigma_{n}^{2}\mathbf{I}_{N}$. The LSEs of $g_{B}$
and $f_{Bm}$ are
\begin{align}
\hat{g}_{B}\!=\!\frac{\boldsymbol{\psi}_{M\!+\!1}^{H}}{\|\boldsymbol{\psi}_{M\!+\!1}\|^{2}}\mathbf{x}_{B},\\
\hat{f}_{Bm}\!=\!\frac{\boldsymbol{\psi}_{m}^{H}}{\|\boldsymbol{\psi}_{m}\|^{2}}\mathbf{x}_{B},
\end{align}
respectively. The mean-square errors (MSEs) of $\hat{g}_{mk}$ and
$\hat{f}_{Bm}$, denoted by $\delta_{mk}$ and $\delta_{Bm}$, are
equal and can be written as
\begin{align}
\delta_{mk}\!=\!\delta_{Bm}\!=\frac{\mathcal{E}\left\{\|\boldsymbol{\psi}^{H}_{m}
\mathbf{w}_{k}\|^{2}\right\}}{\|\boldsymbol{\psi}_{m}\|^{2}}
\!=\!\frac{\mathcal{E}\left\{\|\boldsymbol{\psi}^{H}_{M\!+\!1}
\mathbf{w}_{k}\|^{2}\right\}}{\|\boldsymbol{\psi}_{M\!+\!1}\|^{2}}
\!=\!\underbrace{\frac{\sigma_{n}^{2}}{NP_{s}}}_{\delta_{1}}.
\end{align}

Similarly, the MSEs of $\hat{g}_{B}$ and $\hat{f}_{kM}$, denoted by
$\delta_{B}$ and $\delta_{kM}$, can equivalently be written as
\begin{align}
\delta_{B}\!=\!\delta_{kM}\!=\!\underbrace{\frac{\sigma_{n}^{2}}{NP_{b}}}_{\delta_{2}}.
\end{align}

\subsection{Contract Design under Imperfect CSI}

For notational simplicity, we define $\hat{R}_{C1}$, $\hat{R}_{M1}$,
$\hat{R}_{C2}$, $\hat{R}_{C}$, and $\hat{R}_{M}$ under imperfect CSI
to be compatible with the corresponding data rates $R_{C1}$,
$R_{M1}$, $R_{C2}$, $R_{C}$, and $R_{M}$ under complete CSI. These
transmission data rates under imperfect CSI are characterized in the
following lemma.

\begin{lemma}
The transmission data rates for all RUEs and the MUE at the three
phases are given as follows.

\textbf{Phase I}: The sum transmission data rate of all RUEs and the data rate of the MUE are given by

\begin{align}
&\hat{R}_{C1}\!=\!M\nonumber\\&\log\!\left(1\!+\!
\frac{P_{C1}}{\delta_{1}P_{C1}\!\!\displaystyle\sum_{m\!=\!1}^{M}
\!\!\left[\hat{\mathbf{F}}\hat{\mathbf{F}}^{H}\right]^{-1}_{m,m}\!
\!\!+\!\delta_{2}P_{M1}\left[\hat{\mathbf{F}}\hat{\mathbf{F}}^{H}\right]^{-1}_{M\!+\!1,
M\!+\!1}\!+\!\sigma_{n}^{2}}\right),
\end{align}

\begin{align}
&\hat{R}_{M1}\!=\!\nonumber\\&\log\!\left(1\!+\!
\frac{P_{M1}}{\delta_{1}P_{C1}\!\!\displaystyle\sum_{m\!=\!1}^{M}
\!\!\left[\hat{\mathbf{F}}\hat{\mathbf{F}}^{H}\right]^{-1}_{m,m}\!
\!\!+\!\delta_{2}P_{M}\left[\hat{\mathbf{F}}\hat{\mathbf{F}}^{H}\right]^{-1}_{M\!+\!1,
M\!+\!1}\!+\!\sigma_{n}^{2}}\right),
\end{align}
respectively, where the total transmit power is constrained by
\begin{align}
\sum_{m\!=\!1}^{M}\frac{P_{C1}}{\left(K\!-\!M\!-\!1\right)\left(\upsilon_{m}\!+\!\delta_{1}\right)}\!+\!\frac{P_{M1}}
{\left(K\!-\!M\!-\!1\right)\left(\upsilon_{M\!+\!1}\!+\!\delta_{2}\right)}\!\leq\!P_{\mathrm{max}}.
\end{align}

\textbf{Phase II}: The sum rate of all RUEs can be expressed as
\begin{align}
\hat{R}_{C2}\!=\!M\log\left(1\!+\!\frac{P_{C2}}{\delta_{1}P_{C2}\mathrm{tr}\!
\bigg\{\left[\hat{\mathbf{G}}\hat{\mathbf{G}}^{H}\right]^{-1}\bigg\}\!+\!\sigma_{n}^{2}}\right),
\end{align}
where the total transmit power is constrained by
\begin{align}
\sum_{m\!=\!1}^{M}\frac{P_{C2}}{\left(\upsilon_{m}\!+\!\delta_{1}\right)\left(K\!-\!M\right)}\!\leq\!P_{\mathrm{max}}.
\end{align}

Note that RRHs prefer to achieve the maximal sum rate, and thus
$P_{C2}$ would always be set at the maximum value as
$P_{C2}\!=\!\frac{\left(K\!-\!M\right)P_{\mathrm{max}}}{\varepsilon_{2}}$
with
$\varepsilon_{2}\!=\!\sum_{m=1}^{M}\frac{1}{\upsilon_{m}\!+\!\delta_{1}}$.

\textbf{Phase III}: The sum rate of all RUEs and the data rate of
the MUE are given by
\begin{align}
&\hat{R}_{C3}\nonumber\\&\!=\!\sum_{m\!=\!1}^{M}\log\left(1\!+\!\frac{P_{C3}}{\delta_{1}P_{C3}\mathrm{tr}\!
\bigg\{\left[\hat{\mathbf{G}}\hat{\mathbf{G}}^{H}\right]^{-1}\bigg\}\!+\!|f_{Bm}|^{2}P_{B}\!+\!\sigma_{n}^{2}}\right),
\end{align}
\begin{align}\label{RM3}
\hat{R}_{M3}\!=\!\log\left(1\!+\!\frac{|g_{B}|^{2}P_{B}}{P_{C3}\mathbf{f}\hat{\mathbf{G}}^{H}
\left(\hat{\mathbf{G}}\hat{\mathbf{G}}^{H}\right)^{-2}\hat{\mathbf{G}}\mathbf{f}^{H}
\!+\!\sigma_{n}^{2}}\right),
\end{align}
respectively. With the same power control algorithm as under perfect
CSI, $P_{C3}$ and $P_{B}$ under the fairness power control can be
achieved by
\begin{align}\nonumber
\{P_{C3},{P_B}\}\!&=\!\arg \mathop {\max }\limits_{P_{C3},P_B} \min \left\{ {{\hat{R}_{C3}},{\hat{R}_{M3}}} \right\},\\
s.t.\quad &0\!\leq\!P_{C3}\!\leq\!P_{\rm{max}},\thinspace 0\!\leq\!P_{B}\!\leq\!P_{\rm{max}}.
\end{align}
\end{lemma}
\begin{proof}
See Appendix B.
\end{proof}

The BBU pool can acquire the estimated CSI of the RRH-RUE and
RRH-MUE links, and it cannot estimate CSI of the MBS-MUE link, i.e.,
${\left| {{g_B}} \right|^2}$, which results in information
asymmetry. In order to design the optimal contract, an incentive
compatible (IC) constraint is presented according to the revelation
principle \cite{30}. It is assumed that ${\left|{{g_B}} \right|^2}$
has $L$ potential quantified values and the set containing those
values is denoted by $\Xi = \left\{ {{\xi_1},{\xi _2}, \ldots ,{\xi
_L}} \right\}$. Without loss of generality, we further assume that
${\xi _1} < {\xi _2} < \ldots <{\xi _L}$. The IC constraint can be
defined as follows.

\begin{definition}
(\emph{Incentive Compatible}) A contract is incentive compatible if
the agent with channel coefficients of the MBS-MUE link ${\left|
{{g_B}} \right|^2} = {\xi _l}$ prefers to choose the contract item
$(t_1^l,t_3^l,P_{M1},P_{C1})$ designed specifically for its own
type, i.e.,
\begin{align}
t_1^l\hat R_{M1} + t_3^l\hat R_{M3}^l \ge t_1^{l'}\hat R_{M1} +
t_3^{l'}\hat R_{M3}^l,\forall l,l' \in \left\{ {1,2, \ldots ,L}
\right\}.
\end{align}
\end{definition}

On the other hand, the IR constraint under partial CSI is given by
the following definition.

\begin{definition}
(\emph{Individual Rational}) A contract is individual rational if
the rate-based utility that the agent with ${\left| {{g_B}}
\right|^2} = {\xi_l}$ obtains by choosing the contract item
$(t_1^l,t_3^l,P_{M1},P_{C1})$ is larger than the reservation utility
${u_l}$ , i.e.,
\begin{align}
t_1^l\hat R_{M1} + t_3^l\hat R_{M3}^l \ge {u_l} = {T_0}\hat R_{M3}^l,\forall l \in \left\{ {1,2, \ldots ,L} \right\},
\end{align}
where $t_1^l$ and $t_3^l$ are the time duration of Phase I and Phase
III intended for the agent with ${\left| {{g_B}} \right|^2} =
{\xi_l}$, respectively, and $\hat R_{M3}^l$ is calculated by \eqref{RM3}
under ${\left| {{g_B}} \right|^2} = {\xi _l}$.
\end{definition}

Since the BBU pool can obtain the knowledge of the statistical
parameters and the distributions of random variables by using
appropriate learning and fitting methods, we assume that the BBU
pool can acquire the probability density function (pdf) of
$z\!=\!|f_{Bm}|^{2}$, the set ${\Xi} = \left\{ {{\xi _1},{\xi _2},
\ldots ,{\xi _L}} \right\}$, and the variable ${q_l}$ denoting the
possibility of ${\left| {{g_B}} \right|^2} = {\xi_l}$. Obviously,
${q_l}\in\left[ {0,1} \right]$ and $\sum\limits_{l \in \left\{ {1,2,
\ldots ,L} \right\}} {{q_l}} = 1$.

Hence, when the contract is accepted by the MBS with ${\left|
{{g_B}} \right|^2} = {\xi_l}$, the rate-based utility in the BBU
pool can be written as
\begin{align}\label{U}
{U_C} =& t_1^l\hat R_{C1} + \left( {{T_0} - t_1^l - t_3^l}
\right){{\hat R}_{C2}}
 \nonumber\\&- \left( {{T_0} - t_3^l} \right)\int_z p \left( z \right){{\hat R}_{C3}}dz.
\end{align}

Similarly, the rate-based utility in the MBS can be obtained by
\begin{align}
{U_M} = t_1^l\hat R_{M1} + t_3^l\hat R_{M3}^l.
\end{align}

Then the optimal contract with imperfect CSI can be obtained by
solving the following optimization problem:
\begin{align}
&\mathop {\max }\limits_{\left\{ {\left( {t_1^l,t_3^l,{P_{M1}},{P_{C1}}} \right),\forall l \in \left\{ {1,2, \ldots ,L} \right\}} \right\}} \nonumber\\&
{\rm{ }}\sum\limits_{l \in \left\{ {1,2, \ldots ,L} \right\}}\!\! {{q_l}\left[ {t_1^l{{\hat R}_{C1}} + \left( {{T_0} - t_1^l - t_3^l} \right){{\hat R}_{C2}} - \left( {{T_0} - t_3^l} \right){{\hat R}_{C3}}} \right]}   \\
&s.t.\quad  t_1^l\hat R_{M1} + t_3^l\hat R_{M3}^l \ge {T_0}\hat R_{M3}^l,\forall l \!\in\! \left\{ {1,2, \ldots ,L} \right\},\\
& \quad\quad t_1^l\hat R_{M1} + t_3^l\hat R_{M3}^l \ge t_1^{l'}\hat R_{M1} + t_3^{l'}\hat R_{M3}^l,\nonumber\\ &\quad\quad\forall l,l' \in \left\{ {1,2, \ldots ,L} \right\},\\
& \quad\quad t_3^l \ge 0,t_1^l \ge 0,{T_0} - t_1^l - t_3^l \ge 0,\forall l \in \left\{ {1,2, \ldots ,L} \right\},\\
& \quad\quad {P_{M1}} \ge 0,{P_{C1}} \ge 0,\\
& \quad\quad
\frac{P_{C1}}{\left(K\!-\!M\!-\!1\right)\varepsilon_{2}} +
\frac{{{P_{M1}}}}{\left(K - M - 1\right)(\upsilon
_{M+1}+\delta_{2})}\le P_{\rm{max}}.
\end{align}

The above optimization problem can be further transformed to
\begin{align}
&\mathop {\max }\limits_{\left\{ {\left( {t_3^l,{P_{M1}},{P_{C1}}} \right),\forall l \in \left\{ {1,2, \ldots ,L} \right\}} \right\}} \quad U_{{C_{\left\{ {\left( {t_3^l,{P_{M1}},{P_{C1}}} \right),\forall l \in \left\{ {1,2, \ldots ,L} \right\}} \right\}}}}^*\label{contract}\\
&s.t.\quad 0 \le t_3^1 \le t_3^2 \le  \cdots  \le t_3^L,\label{cons1}\\
& \quad\quad {T_0} - \tilde t_1^l - t_3^l \ge 0,\forall l \in \left\{ {1,2, \ldots ,L} \right\},\\
& \quad\quad{P_{M1}} \ge 0,{P_{C1}} \ge 0,\\
& \quad\quad\frac{P_{C1}}{\left(K\!-\!M\!-\!1\right)\varepsilon_{2}}
+ \frac{{{P_{M1}}}}{\left(K - M - 1\right)(\upsilon
_{M+1}+\delta_{2})}\le P_{\rm{max}},\label{cons2}
\end{align}
where
\begin{align}
&U_{{C_{\left\{ {\left( {t_3^l,{P_{M1}},{P_{C1}}} \right),\forall l
\in \left\{ {1,2, \ldots ,L} \right\}} \right\}}}}^* \nonumber\\&=
\left[ {{{\hat R}_{C1}} - {{\hat R}_{C2}}} \right]\sum\limits_{l \in
\left\{ {1,2, \ldots ,L} \right\}} {{q_l}} \tilde t_1^l - \hat
R\sum\limits_{l \in \left\{ {1,2, \ldots ,L} \right\}} {{q_l}} t_3^l
+ {T_0}\hat R,
\end{align}
\begin{align}
\hat R = {\hat R_{C2}} - {\hat R_{C3}},\thinspace
\tilde t_1^1 = \frac{{{u_L} - t_3^1\hat R_{M3}^1}}{{{{\hat R}_{M1}}}},
\end{align}
\begin{align}
&\tilde t_1^l = \frac{1}{{{{\hat R}_{M1}}}}\left[ {{u_L} - t_3^1\hat
R_{M3}^1 + \sum\limits_{i = 2}^l {\hat R_{M3}^i\left( {t_3^{i - 1} -
t_3^i} \right)} } \right], \nonumber\\&\forall l \in \left\{ {2,
\ldots ,L} \right\}.
\end{align}

Note that the rate-based utility optimization in \eqref{contract}
with constraints \eqref{cons1}--\eqref{cons2} cannot always be
solved in closed-form, which makes it difficult to design the
optimal contract directly. Nevertheless, once the pdf of each random
variable is determined, the optimal contract can be achieved
according to the following Algorithm 2.

\begin{table}
\center \caption{Algorithm 2: one-dimensional search for
$P_{M1}^{*}$ under imperfect CSI}
\begin{tabular}{p{3.4in}@{}}
\hline
\textbf{Input:} $P_{M1}^{(1)}$, $P_{M1}^{(2)}$, where $P_{M1}^{(1)},
P_{M2}^{(2)}\in[0,
(K\!-\!M\!-\!1)(\upsilon_{M+1}+\delta_{2})P_{\rm{max}}]$ satisfy
$P_{M1}^{(2)}>P_{M1}^{(1)}$,\\ \quad $\varsigma_{1}=\frac{\partial
{U}_{C}^{*}} {\partial P_{M1}}\big|_{P_{M1}=P_{M1}^{(1)}}<0$, and
$\varsigma_{2}=\frac{\partial {U}_{C}^{*}}
{\partial P_{M1}}\big|_{P_{M1}=P_{M1}^{(2)}}>0$.\\
\textbf{Initialize:} $U_{C}^{*(1)}$, $U_{C}^{*(2)}$ by substituting $P_{M1}^{(1)}$ and $P_{M1}^{(2)}$ into \eqref{U}, respectively,\\
\enspace $k=1$, error $\epsilon\!>\!0$.\\
\thinspace \textbf{Repeat until convergence:} \\
\enspace Calculate:\\
\enspace
$s=\frac{3[U_{C}^{*(2)}-U_{C}^{*(1)}]}{P_{M1}^{(2)}-P_{M1}^{(1)}}$,\thinspace$z=s-\varsigma_{1}-\varsigma_{2}$,
\thinspace$w=\sqrt{z^{2}-\varsigma_{1}\varsigma_{2}}$,\thinspace$\bar{P}_{M1}=P_{M1}^{(1)}+(P_{M1}^{(2)}
-P_{M1}^{(1)})\left(1-\frac{\varsigma_{2}+w+z}
{\varsigma_{2}-\varsigma_{1}+2w}\right)$.\\
\enspace \textbf{If} $|P_{M1}^{(2)}-P_{M1}^{(1)}|\leq\epsilon$\\
\quad Termination criterion is satisfied.\\
\enspace \textbf{Else}\\
\quad Calculate $\bar{U}_{C}$ and $\varsigma=
\frac{\partial U_{C}^{*}}{\partial P_{M1}}\big|_{P_{M1}=\bar{P}_{M1}}$ by substituting $\bar{P}_{M1}$ into \eqref{U}.\\
\enspace\textbf{If} $\varsigma=0$\\
\quad Termination criterion is satisfied.\\
\enspace\textbf{Else}\\
\quad\textbf{If} $\varsigma<0$\\
\qquad$P_{M1}^{(1)}=\bar{P}_{M1}$, $U_{C}^{*(1)}=\bar{U}_{C}$ and $\varsigma_{1}=\varsigma$,\\
\quad\textbf{If} $\varsigma>0$\\
\qquad$P_{M1}^{(2)}=\bar{P}_{M1}$, $U_{C}^{*(2)}=\bar{U}_{C}$ and $\varsigma_{2}=\varsigma$,\\
\enspace$k\!=\!k\!+\!1$.\\\hline
\end{tabular}
\end{table}

\textit{Remark}: The computational complexity of Algorithm 2 is
similar to that of Algorithm 1.

\subsection{Discussion with the Sufficiently Large $K$}

In this subsection, we consider the case in which the number of RRHs
is large. It is assumed that each radio channel coefficient has a
complex Gaussian distribution with mean zero. According to the Law
of Large Numbers, we have the following approximations:
\begin{align}\label{large}
\frac{1}{K}\hat{\mathbf{G}}\hat{\mathbf{G}}^{H}\!\rightarrow\!\mathbf{D}\!+\!\delta_{1}\mathbf{I}_{M},
\thinspace\frac{1}{K}\hat{\mathbf{F}}\hat{\mathbf{F}}^{H}\!\rightarrow\!\mathrm{diag}\left[
\mathbf{D}\!+\!\delta_{1}\mathbf{I}_{M}, \left(\upsilon_{M+1}\!+\!\delta_{2}\right)\right].
\end{align}

\begin{definition}
When the number of RRHs is sufficiently large, the transmission data
rates of all RUEs and the MUE at three phases under imperfect CSI
can be rewritten as
\begin{align}
\tilde{R}_{C1}\!&=\!M\log\left(1\!+\!\frac{P_{C1}}{\frac{\delta_{1}P_{C1}}
{K\varepsilon_{2}}\!+\!\frac{\delta_{2}P_{M}}
{K\left(\upsilon_{M+1}\!+\!\delta_{2}\right)}\!+\!\sigma_{n}^{2}}\right),\label{RC1}\\
\tilde{R}_{M1}\!&=\!\log\left(1\!+\!\frac{P_{M1}}{\frac{\delta_{1}P_{C1}}
{K\varepsilon_{2}}\!+\!\frac{\delta_{2}P_{M1}}
{K\left(\upsilon_{M+1}\!+\!\delta_{2}\right)}\!+\!\sigma_{n}^{2}}\right),\\
\tilde{R}_{C2}\!&=\!M\log\left(1\!+\!\frac{P_{C2}}{\frac{\delta_{1}P_{C2}}
{K\varepsilon_{2}}\!+\!\sigma_{n}^{2}}\right),\\
\tilde{R}_{C3}\!&=\!\sum_{m\!=\!1}^{M}\log\left(1\!+\!\frac{P_{C3}}{\frac{\delta_{1}P_{C3}}
{K\varepsilon_{2}}\!+\!|f_{Bm}|^{2}P_{B}\!+\!\sigma_{n}^{2}}\right),\\
\tilde{R}_{M3}\!&=\!\log\left(1\!+\!\frac{|g_{B}|^{2}P_{B}}{\frac{\upsilon_{M+1}P_{C3}}{K\varepsilon_{2}
}\!+\!\sigma_{n}^{2}}\right)\label{RMM3}.
\end{align}
\end{definition}

Eqs. \eqref{RC1}--\eqref{RMM3} in \textbf{Definition 4} can be
derived directly by substituting \eqref{large} into the transmission
data rates in \textbf{Lemma 2}. Based on \textbf{Definition 4}, the
optimal contract design can be derived via the following lemma when
the number of RRHs is sufficiently large.

\begin{lemma}
For a sufficiently large $K$ under imperfect CSI, once the contract
is accepted by the agent, the optimal $t_{3}^{l}$, $P_{C1}^{*}$, and
$P_{M1}^{*}$ follow
\begin{align}
&t_{3}^{l}\!=\!0,l=1,\ldots,L-1,\\
&t_{3}^{L}\!=\!
\begin{cases}
T_{0}& {{\tilde R}_{M1}}({{\tilde R}_{C2}} - {{\tilde R}_{C3}}) < \tilde R_{M3}^L({{\tilde R}_{C2}} - {{\tilde R}_{C1}}),\\
[0,T_{0}]&{{\tilde R}_{M1}}({{\tilde R}_{C2}} - {{\tilde R}_{C3}}) = \tilde R_{M3}^L({{\tilde R}_{C2}} - {{\tilde R}_{C1}}),\\
0&{{\tilde R}_{M1}}({{\tilde R}_{C2}} - {{\tilde R}_{C3}}) > \tilde R_{M3}^L({{\tilde R}_{C2}} - {{\tilde R}_{C1}}),
\end{cases}\\
&\frac{P_{C1}^{*}}{\left(K\!-\!M\!-\!1\right)\varepsilon_{2}}
\!+\!\frac{P_{M1}^{*}}{\left(K\!-\!M\!-\!1\right)(\upsilon_{M\!+\!1
}\!+\!\delta_{2})}\!=\!P_{\mathrm{max}}.
\end{align}
\end{lemma}
\begin{proof}
See Appendix C.
\end{proof}

\section{Numerical Results}

In this section, Monte Carlo simulations are provided to evaluate
the performances of the proposed contract-based interference
coordination framework in H-CRANs numerically. All the fast fading
channel coefficients are generated as independent complex Gaussian
random variables with zero means and unit variances. The large-scale
fading coefficients $\upsilon_{m}$ ($m=1,\ldots,(M+1)$) are modeled
as $\upsilon_{m}\!=\!z_{m}/(r_{m}/r_{0})^{v}$, where $z_{m}$ is
log-normal random variable with standard deviation
$8\thinspace\textup{dB}$, $r_{m}$ is the distance between the $m$-th
UE and RRHs, $r_{0}$ is set to be $100$ meters, and $v\!=\!3.8$ is
the path loss exponent. The other common parameters are set at
$P_{\mathrm{max}}\!=\!1$, $T_{0}\!=\!1$, $M=4$ and $N\!=\!10$; thus
the signal-to-noise ratio (SNR) is given by
$\textup{SNR}\!=\!\frac{P_{\rm{max}}}{\sigma_{n}^{2}}
\!=\!\frac{1}{\sigma_{n}^{2}}$. Totally $10^{5}$ Monte-Carlo runs
are used for the average calculations. Two baselines are considered
in this paper:

\begin{itemize}
\item \emph{Frequency reuse with power control} (FRPC): RRHs and the MBS reuse the same RB, and the transmit powers are determined
according to the fairness power control adopted for the proposed
CICF in Phase III:
    \begin{align}\nonumber
    \{P_{C3},{P_B}\}\!&=\!\arg \mathop {\max }\limits_{P_{C3},P_B} \min \left\{ {{R_{C3}},{R_{M3}}} \right\},\\
    s.t.\quad &0\!\leq\!P_{C3}\!\leq\!P_{\rm{max}},\thinspace 0\!\leq\!P_{B}\!\leq\!P_{\rm{max}}.\nonumber
    \end{align}
\item \textit{Time domain interference cancelation} (TDIC): RRHs and the MBS transmit separately in the different time phases, where the transmission duration allocated to RRHs and the MBS are $t_{R}$ and $t_{M}$, respectively, and they satisfy the relationship with $t_{R}\!=\!t_{M}\!=\!\frac{T_{0}}{2}$.
\end{itemize}

\subsection{Performance Evaluations under Perfect CSI}

In Fig. \ref{simu1}, the sum rates of all RUEs versus SNRs under
perfect CSI are compared among the CICF, FRPC and TDIC. According to
these simulation curves, the proposed CICF achieves a larger sum
rate than the two baselines do, which indicates that the BBU pool as
the principal can benefit from the proposed contact model, thus
demonstrating the effectiveness of the proposed CICF for all RUEs.
Due to gains from the optimized time duration of the three phases
and the allocated transmission power of RRHs and the MBS, the sum
rate of all RUEs increases almost linearly with SNRs, whereas the
baseline FRPC cannot provide the same performance gain as the CICF
because the intr-tier interference is still severe even though the
fairness power control is utilized. For the baseline TDIC, the sum
rate is the lowest because its spectral efficiency is low even
though the inter-tier interference is avoided.

\begin{figure}[!h!t]
\center
 \includegraphics[width=0.45\textwidth]{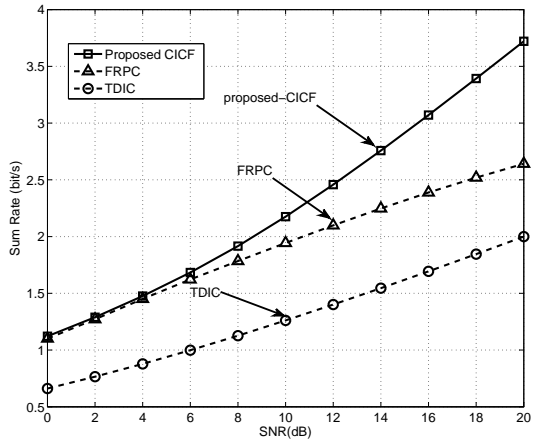}
 \caption{Sum rates versus SNRs for all RUEs with different schemes.} \label{simu1}
\end{figure}

To evaluate performance gains of the agent, the data rates of the
MUE versus SNRs are compared among the CICF, FRPC and TDIC in Fig.
\ref{simu2}. It can be observed that the proposed CICF achieves the
equivalent data rate for the MUE as the FRPC, which is not
unexpected since the equality of the IR constraint would always hold
when the BBU pool can acquire the complete and perfect CSI.
Comparing with the baseline TDIC, the data-rate based utility gain
for the MUE is significant for the proposed CICF because only
partial spectral resources are utilized in the TDIC.

\begin{figure}[!h!t]
\center
 \includegraphics[width=0.45\textwidth]{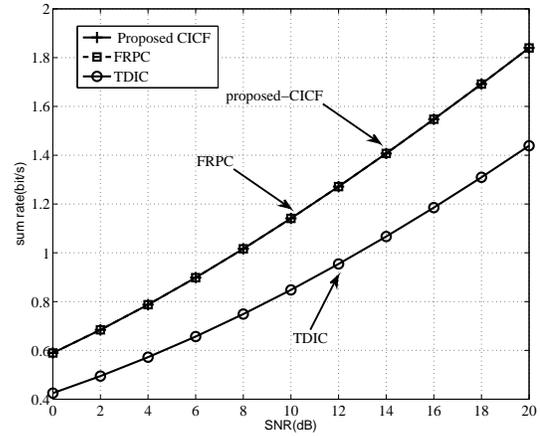}
 \caption{Data rates versus SNRs for MUEs with different schemes.} \label{simu2}
\end{figure}

Summarizing the sum rates of all RUEs and MUE in Fig. \ref{simu1}
and Fig. \ref{simu2}, respectively, the rate-based utility gain from
the proposed CICF is significant, which demonstrating the
effectiveness of the proposed contract-based solution under perfect
CSI.

\subsection{Performance Evaluations under Partial CSI}

In Fig. \ref{simu3}, the sum rates of all RUEs versus SNRs under
imperfect CSI for the proposed CICF and baselines (i.e., FRPC and
TDIC) are evaluated and compared. Similarly to the results under
perfect CSI, the proposed CICF can achieve a significant sum rate
gain compared with the baselines due to the fact that both the time
durations of the three phases and the allocated power are optimized.
Moreover, comparing the performance gain gap from the CICF to FRPC
in Fig. \ref{simu1} and Fig. \ref{simu3}, the benefit of the
principal under imperfect CSI is lower than that under perfect CSI.
This happens because the derived optimal contract under imperfect
CSI cannot always ensure the equality of all IR constraints.
Furthermore, the performances of FRPC and TDIC under perfect CSI is
better than that under imperfect CSI due to the fact that the
channel estimation error degrades the precoding performance at RRHs.

\begin{figure}[!h!t]
\center
 \includegraphics[width=0.45\textwidth]{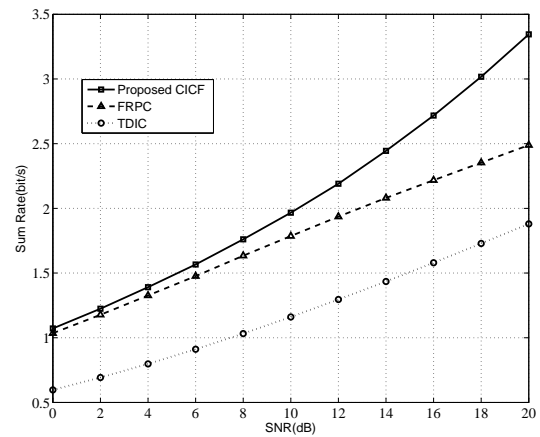}
 \caption{Sum rates versus SNRs for RUEs with different schemes.} \label{simu3}
\end{figure}

In Fig. \ref{simu4}, data rates of the MUE versus SNRs under
imperfect CSI are compared among the CICF, FRPC and TDIC. Comparing
with FRPC, the agent can achieve a significant performance gain in
the proposed CICF because the optimal contract under imperfect CSI
cannot ensure the equality of all IR constraints. Meanwhile, the
performance of FRPC in Fig. \ref{simu4} is worse than that in Fig.
\ref{simu2} because the fairness power control depends critically on
the instantaneous CSI. While TDIC in Fig. \ref{simu4} is the same as
that in Fig. \ref{simu2} since the CSI has no impact on the downlink
transmission of the MBS with this baseline.

\begin{figure}[!h!t]
\center
 \includegraphics[width=0.45\textwidth]{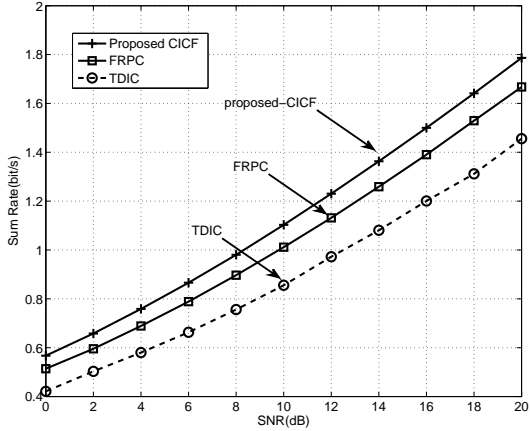}
 \caption{Data rates versus SNRs for MUEs with different schemes.} \label{simu4}
\end{figure}

To characterize the optimal contract in the proposed CICF, the
cumulative density function (cdf) of both $t_{1}^{*}$ and
$t_{2}^{*}$ are plotted in Fig. \ref{simu6}. It can be seen that
$P\!\left(t_{1}^{*}\!\geq\!0.5T_{0}\right)\!>\!0.7$ and
$P\!\left(t_{2}^{*}\!\leq\!0.5T_{0}\right)\!>\!0.8$, which indicates
that the time duration of Phase I should be longer than that of
Phase II to optimize the sum date rates. To optimize the
transmission data rates of the whole H-CRAN, more time resources
should be assigned to allow MUEs to be served by RRHs, and thus an
example of the transmission frame structure under these scheduling
algorithms and configurations in this paper is suggested such that
the time duration of Phase I is about 70 percent of the TTI length,
and the time duration of Phase II is about 30 percent of the TTI
length. Note that the transmission duration of Phase III would be
zero when the contract is signed, which is indicated in both Lemma 1
and Lemma 2.

\begin{figure}[!h!t]
\center
 \includegraphics[width=0.45\textwidth]{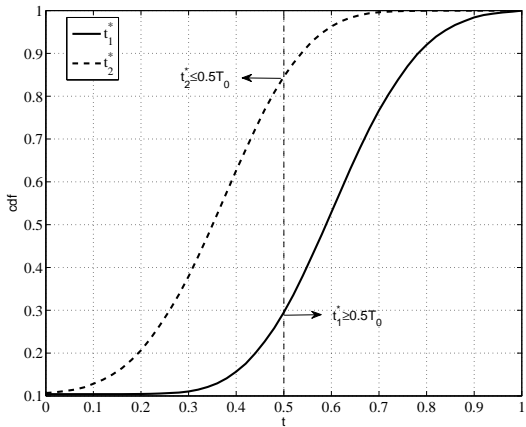}
 \caption{The cdf distributions of $t_{1}^{*}$ and $t_{2}^{*}$.} \label{simu6}
\end{figure}

\section{Conclusions}

In this paper, a contract-based interference coordination framework
has been proposed to mitigate inter-tier interference in H-CRANs,
whose core idea is that the BBU pool is selected as the principal
that would present a contract to the MBS, and the MBS as the agent
decides whether to accept the contract or not according to its
potential for improving the SE performance. An optimal contract
design that maximizes the overall rate-based utility has been
derived when perfect CSI is available at both principal and agent.
Furthermore, the contract optimization under the situation where the
partial CSI has been obtained from channel estimation has been
considered as well. Monte Carlo simulations have been provided to
corroborate the analysis, and simulation results have shown that the
proposed framework can significantly increase the sum data rates,
thus demonstrating the effectiveness of the proposed contract-based
framework. Interesting topics for further research include, more
advanced inter-tier interference coordination schemes to enhance the
\emph{RRH-MBS with UEs-separated} phase in the presented framework,
and the inclusion of fronthaul constraints within this framework.

\appendices
\section{Proof of Lemma 1}

For the fixed $t_{1}$ and $t_{2}$, the optimal $P_{C1}$ and $P_{M1}$
can be derived by using the Karush-Kuhn-Tucker (KKT) conditions. The
corresponding Lagrangian function can be written as
\begin{align}
L&(P_{C1}, P_{M1}, \lambda)\!=\!-t_{1}R_{C1}
\nonumber\\&\!-\!\lambda\left[P_{\mathrm{max}}\!-\!\frac{P_{C1}}{(K\!-\!M\!-\!1)\varepsilon_{1}}\!-\!\frac{P_{M1}}
{(K\!-\!M\!-\!1)\upsilon_{M\!+\!1}}\right],
\end{align}
whose solution is derived as
\begin{align}
&\nabla L(P_{C1}, P_{M1}, \lambda)\!=-t_{1}\frac{\partial R_{C1}}{\partial P_{C1}}\!
+\!\lambda\frac{1}{(K\!-\!M\!-\!1)\varepsilon_{1}}\!=\!0,\\
&\lambda\left[P_{\rm{max}}\!-\!\frac{P_{C1}}
{(K\!-\!M\!-\!1)\varepsilon_{1}}\!-\!\frac{P_{M1}}{(K\!-\!M\!-\!1)
\upsilon_{M\!+\!1}}\right]=0,\\
&\lambda\geq 0.
\end{align}

It is obvious that $\lambda$ must be positive to achieve the optimal
solution for $P_{C1}$, and thus the optimal received powers
$P_{C1}^{*}$ and $P_{M1}^{*}$ allocated to RUEs and MUE should
satisfy
\begin{align}
\frac{P_{C1}^{*}}{(K\!-\!M\!-\!1)\varepsilon_{1}}\!+\!\frac{P_{M1}^{*}}
{(K\!-\!M\!-\!1)\upsilon_{M\!+\!1}}\!=\!P_{\mathrm{max}}.
\end{align}

For given $P_{C1}$ and $P_{M1}$, the optimal $t_{1}$ and $t_{2}$ can
be obtained by solving following equations:
\begin{align}
&R_{C2}\!-\!R_{C1}\!+\!v\!-\!v_{1}\!-\!wR_{M1}\!=\!0,\\
&R_{C2}\!-\!R_{C3}\!+\!v\!-\!v_{2}\!-\!wR_{M3}\!=\!0,\\
&v_{1}t_{1}\!=\!0,\thinspace v_{2}t_{2}\!=\!0,\thinspace v(T\!-\!t_{1}\!-\!t_{2})\!=\!0,\\
&\varpi(t_{1}R_{M1}\!+\!t_{3}R_{M3}\!-\!u)\!=\!0,\\
&v_{1}\!\geq\!0,\thinspace v_{2}\!\geq\!0,\thinspace
v\!\geq\!0,\thinspace \varpi\!\geq\!0.
\end{align}
where $v$, $\varpi$, $v_{1}$, and $v_{2}$ are the Lagrangian multipliers.
It is obvious that $\varpi$ must be larger than zero. Therefore, the IR
constraint is derived. Moreover, in order to maximize the
transmission bit rate of the MUE, the optimal $t_{3}$ should be
equal to zero.

\section{Proof of Lemma 2}

Denoting the $m$-th column of $\mathbf{G}^{T}$ and $\mathbf{A}$ by
$\mathbf{g}_{m}$ and $\mathbf{a}_{m}$, respectively, the
transmission data rate can be derived as \cite{31}
\begin{align}
&\hat{R}_{m}\!=\!\nonumber\\&\log\!\left(1\!+\!\frac{|\mathcal{E}[\mathbf{g}_{m}\mathbf{a}_{m}]|^{2}}
{\sigma_{n}^{2}\!+\!\textup{var}\{\mathbf{g}_{m}\mathbf{a}_{m}\}\!+\!|f_{Bm}|^{2}P_{B}\!
+\!\sum_{i\neq m}^{M}
\mathcal{E}[|\mathbf{g}_{m}\mathbf{a}_{i}|^{2}]}\right).
\end{align}

For the ZF precoder, we have $\mathbf{\hat{g}}_{m}\mathbf{a}_{i}\! =
\!\delta_{mi}$. Denote the radio channel matrix by $\mathcal{G}\! =
\!\hat{\mathbf{G}}\! + \!\tilde{\mathbf{G}}$, where
$\tilde{\mathbf{G}}$ is the estimation error, and we have
\begin{align}\label{app_b_1}
\mathbf{g}_{m}\mathbf{a}_{i}\! =
\!\hat{\mathbf{g}}_{m}\mathbf{a}_{i}
+\!\tilde{\mathbf{g}}_{m}\mathbf{a}_{i} =
\!\delta_{mi}\!+\!\tilde{\mathbf{g}}_{m}\mathbf{a}_{i}.
\end{align}

According to (\ref{app_b_1}), the expectation and variance can be
written as
\begin{align}
\mathcal{E}[\mathbf{g}_{m}\mathbf{a}_{i}]\!&=\!P_{C1},
\thinspace \mathcal{E}[|\mathbf{g}_{m}\mathbf{a}_{i}|^{2}]\!=\!\delta_{1}P_{C1}
\sum_{i\neq m}^{M}\left[\mathbf{G}\mathbf{G}^{H}\right]^{-1}_{ii},\\
\textup{var}\{\mathbf{g}_{m}\mathbf{a}_{i}\}\!&=\!
\mathcal{E}\left[\mathbf{a}_{m}^{H}\mathbf{g}_{m}^{H}\mathbf{g}_{m}
\mathbf{a}_{m}\right]\!=\!\delta_{1}P_{C1}\left[\mathbf{G}\mathbf{G}^{H}\right]^{-1}_{mm}.
\end{align}

Thus, the sum data rate of all RUEs in the downlink transmission
during Phase III is given by
\begin{align}
\hat{R}_{C3}\!=\!\sum_{m\!=\!1}^{M}\!\log\!\!\left(1\!+\!\frac{P_{C3}}{\delta_{1}P_{C3}\mathrm{tr}\!
\bigg\{\left[\hat{\mathbf{G}}\hat{\mathbf{G}}^{H}\right]^{-1}\bigg\}\!
+\!|f_{Bm}|^{2}P_{B}\!+\!\sigma_{n}^{2}}\right),
\end{align}
and the data rate of the MUE can be written as
\begin{align}
\hat{R}_{M3}\!&=\!\log\left(1\!+\!\frac{|g_{B}|^{2}P_{B}}{\sigma_{n}^{2}\!
+\!|\mathbf{fA}|^{2}}\right)
\nonumber\\&=\!\log\left(1\!+\!\frac{|g_{B}|^{2}P_{B}}{P_{C3}\mathbf{f}\hat{\mathbf{G}}^{H}
\left(\hat{\mathbf{G}}\hat{\mathbf{G}}^{H}\right)^{-2}\hat{\mathbf{G}}\mathbf{f}^{H}
\!+\!\sigma_{n}^{2}}\right).
\end{align}

Similarly, the sum rates during the Phase I and II can be calculated
in the same way.

\section{Proof of Lemma 3}

By using the KKT conditions, the optimal $\tilde{t}_{3}^{L}$ can be
obtained by solving the following equations:
\begin{align}
&-(\tilde{R}_{C2}\!-\!\tilde{R}_{C1})q_{L}\frac{\tilde{R}_{M3}^{L}}
{\tilde{R}_{M1}}\nonumber\\&\quad
\quad\!+\!(\tilde{R}_{C2}\!-\!\tilde{R}_{C3})q_{L}\!-\!\omega_{1}
\!+\!\omega_{2}(1\!-\!\frac{\tilde{R}_{M3}^{L}}{\tilde{R}_{M1}})=0,\label{app_3_1}\\
&\omega_{1}t_{3}^{L}\!=\!0,\thinspace\omega_{2}(T_{0}
\!-\!\tilde{t}_{1}^{L}\!-\!t_{3}^{L})\!=\!0,\label{app_3_2}\\
&\omega_{1}\!\geq\!0,\thinspace\omega_{2}\!\geq\!0.\label{app_3_3}
\end{align}

Based on (\ref{app_3_1}), (\ref{app_3_2}), and (\ref{app_3_3}), we
have the following results:
\begin{itemize}
\item When $\tilde{R}_{M1}(\tilde{R}_{C2}-\tilde{R}_{C3})\!> \!\tilde{R}_{M3}^{L}(\tilde{R}_{C2}\!-\!\tilde{R}_{C1})$, $\omega_{1}$ must be positive, i.e., $t_{3}^{L}\!=\!0$. As $t_3^l$ ($l = 1, \ldots ,L - 1$) is often less than $t_{3}^{L}$, $t_3^l$ is equal to zero.
\item When $\tilde{R}_{M1}(\tilde{R}_{C2}-\tilde{R}_{C3})\!=\!\tilde{R}_{M3}^{L}(\tilde{R}_{C2}\!-\!\tilde{R}_{C1})$, $t_{3}^{L}$ can be an arbitrary value in the interval $[0,T_{0}]$. The derivative of $-U_{C}^{*}$ with respect to $t_{3}^{L\!-\!1}$ can be written
as
    \begin{align}
    (\tilde{R}_{C2}\!-\!\tilde{R}_{C1})\left[q_{L}\!+\!q
    _{L\!-\!1}\right](\tilde{R}_{M3}^{L}-\tilde{R}_{M3}^{L-1})>0.
    \end{align}
    To maximize $U_{C}^{*}$, we have $t_{3}^{L\!-\!1}\!=\!0$, and all $t_3^l$ ($l < L - 1$) should be equal to zero.
\item When $\tilde{R}_{M1}(\tilde{R}_{C2}-\tilde{R}_{C3})\!< \!\tilde{R}_{M3}^{L}(\tilde{R}_{C2}\! -\!\tilde{R}_{C1})$, $\omega_{2}$ must be positive. Meanwhile, the following equation holds:
    \begin{align}
    T_{0}\!=\!\tilde{t}_{1}^{L}\!+\!t_{3}^{L}.
    \end{align}

    Obviously, this contract is unacceptable for the agent. Therefore,
    $t_{3}^{L}$ should be equal to $T_{0}$. In order to satisfy the IR constraint, $t_3^l$ ($l = 1, \ldots ,L - 1$) should be set as zero.
\end{itemize}

As for the optimal powers assigned to RRHs and MBS in Phase I, when
$t_3^l$ $(l = 1, \ldots ,L - 1)$ is equal to zero, $U_{C}^{*}$ can
be expressed as
\begin{align}
U_{C}^{*}\!=\!\frac{(\tilde{R}_{C1}\!-\!\tilde{R}_{C2})u_{L}}{\tilde{R}_{M1}}\!+\!T_{0}(\tilde{R}_{C3}\!-\!\tilde{R}_{C2}).
\end{align}

The Lagrangian function can be written as
\begin{align}
L_{1}&(P_{C1}, P_{M1},
\lambda_{1})\!=\!-\underbrace{\frac{(\tilde{R}_{C1}\!-\!\tilde{R}_{C2})u_{L}}{\tilde{R}_{M1}}}_{f_{1}(P_{C1},P_{M1})}
\nonumber\\&\!-\!\lambda_{1}\left[P_{\mathrm{max}}\!-\!\frac{P_{C1}}{(K\!-\!M\!-\!1)\varepsilon_{2}}\!-\!\frac{P_{M1}}
{(K\!-\!M\!-\!1)(\upsilon_{M\!+\!1}+\delta_{2})}\right].
\end{align}

The derivative of the Lagrangian function with respect to $P_{C1}$
is
\begin{align}
&\nabla L_{1}(P_{C1}, P_{M1},
\lambda_{1})\!=\nonumber\\&\quad\quad-\frac{\partial
f_{1}(P_{C1},P_{M1})} {\partial P_{C1}}\!
+\!\lambda_{1}\frac{1}{(K\!-\!M\!-\!1)\varepsilon_{2}}\!=\!0,\\
&\lambda_{1}\left[P_{max}\!-\!\frac{P_{C1}}
{(K\!-\!M\!-\!1)\varepsilon_{2}}\!-\!\frac{P_{M1}}{(K\!-\!M\!-\!1)
(\upsilon_{M\!+\!1}+\delta_{2})}\right]=0.\\
&\lambda_{1}\geq 0.
\end{align}

Since the derivative of $f_{1}(P_{C1},P_{M1})$ with respect to
$P_{C1}$ is positive, $\lambda^{*}$ must be positive, and the
optimal powers $P_{C1}^{*}$ and $P_{M1}^{*}$ should satisfy
\begin{align}\label{power1}
\frac{P_{C1}^{*}}{(K\!-\!M\!-\!1)\varepsilon_{2}}\!+\!\frac{P_{M1}^{*}}
{(K\!-\!M\!-\!1)(\upsilon_{M\!+\!1}\!+\!\delta_{2})}\!=\!P_{\mathrm{max}}.
\end{align}

If $t_{3}^{L}$ is positive, $U_{C}^{*}$ can be re-written as
\begin{align}
U_{C}^{*}\!&=\!(\tilde{R}_{C1}\!-\!\tilde{R}_{C2})[1-q_{L}]\frac{u_{L}}{\tilde{R}_{M1}}\!
\nonumber\\&+\!(\tilde{R}_{C1}\!-\!\tilde{R}_{C2})q_{L}\frac{u_{L}\!-\!\tilde{R}_{M3}^{L}}{\tilde{R}_{M1}}
+\!T_{0}(\tilde{R}_{C3}\!-\!\tilde{R}_{C2})[1-q_{L}].
\end{align}

Similarly, the Lagrangian function is re-formulated as
\begin{align}
&L_{2}(P_{C1}, P_{M1},
\lambda_{2})\!=\!\nonumber\\&-\underbrace{(\tilde{R}_{C1}\!-\!\tilde{R}_{C2})[1-q_{L}]\frac{u_{L}}
{\tilde{R}_{M1}}\!+\!(\tilde{R}_{C1}\!-\!\tilde{R}_{C2})q_{L}\frac{u_{L}
\!-\!\tilde{R}_{M3}^{L}}{\tilde{R}_{M1}}}_{f_{2}(P_{C1},P_{M1})}
\!\\
&-\!\lambda_{2}\left[P_{\mathrm{max}}\!-\!\frac{P_{C1}}{(K\!-\!M\!-\!1)\varepsilon_{2}}\!-\!\frac{P_{M1}}
{(K\!-\!M\!-\!1)(\upsilon_{M\!+\!1}+\delta_{2})}\right].
\end{align}

As the derivative of $f_{2}(P_{C1},P_{M1})$ with respect to $P_{C1}$
is also positive, the optimal powers allocated to $P_{C1}^{*}$ and
$P_{M1}^{*}$ are the same as in \eqref{power1}.

\begin{IEEEbiography}{Mugen Peng}
(M'05--SM'11) received the B.E. degree in Electronics Engineering
from Nanjing University of Posts \& Telecommunications, China in
2000 and a PhD degree in Communication and Information System from
the Beijing University of Posts \& Telecommunications (BUPT), China
in 2005. After the PhD graduation, he joined in BUPT, and has become
a full professor with the school of information and communication
engineering in BUPT since Oct. 2012. During 2014, he is also an
academic visiting fellow in Princeton University, USA. He is leading
a research group focusing on wireless transmission and networking
technologies in the Key Laboratory of Universal Wireless
Communications (Ministry of Education) at BUPT. His main research
areas include wireless communication theory, radio signal processing
and convex optimizations, with particular interests in cooperative
communication, radio network coding, self-organizing network,
heterogeneous network, and cloud communication. He has
authored/coauthored over 50 refereed IEEE journal papers and over
200 conference proceeding papers.

Dr. Peng is currently on the Editorial/Associate Editorial Board of
the \emph{IEEE Communications Magazine}, the IEEE Access, the
\emph{IET Communications}, the \emph{International Journal of
Antennas and Propagation} (IJAP), the \emph{China Communications},
and the \emph{International Journal of Communications System}
(IJCS). He has been the guest leading editor for the special issues
in the \emph{IEEE Wireless Communications}. Dr. Peng was a recipient
of the 2014 IEEE ComSoc AP Outstanding Young Researcher Award, and
the best paper award in GameNets 2014, CIT 2014, ICCTA 2011, IC-BNMT
2010, and IET CCWMC 2009. He received the First Grade Award of
Technological Invention Award in Ministry of Education of China for
his excellent research work on the hierarchical cooperative
communication theory and technologies, and the Second Grade Award of
Scientific and Technical Advancement from China Institute of
Communications for his excellent research work on the co-existence
of multi-radio access networks and the 3G spectrum management.
\end{IEEEbiography}

\begin{IEEEbiography}{Xinqian Xie}
received the B.S. degree in telecommunication engineering from the
Beijing University of Posts and Communications (BUPT), China, in
2010. He is currently pursuing the Ph.D. degree at BUPT. His
research interests include cooperative communications, estimation
and detection theory.
\end{IEEEbiography}

\begin{IEEEbiography}{Qiang Hu}
received the B.S. degree in applied physics from the Beijing
University of Posts \text{\&} Communications (BUPT), China, in 2013.
He is currently pursuing the Master degree at BUPT. His research
interests include cooperative communications, such as cloud radio
access networks (C-RANs), and statistical signal processing in
large-scale networks.
\end{IEEEbiography}

\begin{IEEEbiography}{Jie
Zhang} (M'02) is a full professor and has held the Chair in Wireless
Systems at the Electronic and Electrical Engineering Dept.,
University of Sheffield since 2011. He received PhD in 1995 and
became a Lecturer, Reader and Professor in 2002, 2005 and 2006
respectively. He and his students pioneered research in femto/small
cell and HetNets and published some of the most widely cited
publications in these topics. He co-founded RANPLAN Wireless Network
Design Ltd. (www.ranplan.co.uk) that produces a suite of world
leading in-building distributed antenna system, indoor-outdoor small
cell/HetNet network design and optimization tools
``iBuildNet-Professional, Tablet and Cloud''. Since 2003, he has
been awarded over 20 projects by the EPSRC, the EC FP6/FP7/H2020 and
industry, including some of earliest research projects on
femtocell/HetNets.
\end{IEEEbiography}

\begin{IEEEbiography}{H. Vincent Poor}
(S'72, M'77, SM'82, F'87) received the Ph.D. degree in EECS from
Princeton University in 1977.  From 1977 until 1990, he was on the
faculty of the University of Illinois at Urbana-Champaign. Since
1990 he has been on the faculty at Princeton, where he is the
Michael Henry Strater University Professor of Electrical Engineering
and Dean of the School of Engineering and Applied Science. Dr.
Poor's research interests are in the areas of information theory,
statistical signal processing and stochastic analysis, and their
applications in wireless networks and related fields including
social networks and smart grid. Among his publications in these
areas are the recent books \textit{Principles of Cognitive Radio}
(Cambridge University Press, 2013) and \textit{Mechanisms and Games
for Dynamic Spectrum Allocation} (Cambridge University Press, 2014).

Dr. Poor is a member of the National Academy of Engineering, the
National Academy of Sciences, and is a foreign member of Academia
Europaea and the Royal Society. He is also a fellow of the American
Academy of Arts and Sciences, the Royal Academy of Engineering (U.
K), and the Royal Society of Edinburgh. He received the Marconi and
Armstrong Awards of the IEEE Communications Society in 2007 and
2009, respectively. Recent recognition of his work includes the 2014
URSI Booker Gold Medal, and honorary doctorates from Aalborg
University, Aalto University, the Hong Kong University of Science
and Technology, and the University of Edinburgh.

\end{IEEEbiography}


\begin{thebibliography}{1}

\bibitem{1}
M. Peng and W. Wang, \textquotedblleft Technologies and standards
for TD-SCDMA evolutions to
IMT-Advanced,\textquotedblright~\emph{IEEE Commun. Mag.}, vol. 47,
no. 12, pp. 50--58, Dec. 2009.

\bibitem{2}
M. Peng, Y. Liu, D. Wei, W. Wang, and H. Chen, \textquotedblleft
Hierarchical cooperative relay based heterogeneous
networks,\textquotedblright~\emph{IEEE Wireless Commun.}, vol. 18,
no. 3, pp. 48--56, Jun. 2011.

\bibitem{3}
M. Peng, D. Liang, Y. Wei, J. Li, and H. Chen, \textquotedblleft
Self-configuration and self-optimization in LTE-Advanced
heterogeneous networks,\textquotedblright~\emph{IEEE Commun. Mag.},
vol. 51, no. 5, pp. 36--45, May 2013.

\bibitem{4}
W. Shin, W. Noh, K. Jang, and H. Choi, \textquotedblleft
Hierarchical interference alignment for downlink heterogeneous
networks,\textquotedblright~\emph{IEEE Trans. Wireless Commun.},
vol. 11, no. 12, pp. 4549--4559, Dec. 2012.

\bibitem{5}
M. Peng, X. Zhang, and W. Wang, \textquotedblleft Performance of
orthogonal and co-channel resource assignments for femto-cells in
LTE systems,\textquotedblright~\emph{IET Communications}, vol. 7,
no. 5, pp. 996--1005, May 2011.

\bibitem{6}
P. Xia, C. Liu, and J. G. Andrews, \textquotedblleft Downlink
coordinated multi-point with overhead modeling in heterogeneous
cellular networks,\textquotedblright~\emph{IEEE Trans. Wireless
Commun.}, vol. 12, no. 8, pp. 4025--4037, Aug. 2013.

\bibitem{7}
R. Irmer, H. Droste, P. Marsch \emph{et al.}, \textquotedblleft
Coordinated multipoint: Concepts, performance, and field trial
results,\textquotedblright~\emph{IEEE Commun. Mag.}, vol. 49, no. 2,
pp. 102--111, Feb. 2011.

\bibitem{8}
M. Peng, C. I, C. Tan, and C. Huang, \textquotedblleft IEEE access
special section editorial: Recent advances in cloud radio access
networks,\textquotedblright~\emph{IEEE Access}, vol. 2, pp.
1683--1685, Dec. 2014.

\bibitem{9}
X. Xie, M. Peng, W.Wang, and H. V. Poor, \textquotedblleft Training
design and channel estimation in uplink cloud radio access
networks,\textquotedblright~\emph{IEEE Signal Process. Letters},
vol. 22, no. 8, pp. 1060--1064, Aug. 2015.

\bibitem{10}
M. Peng, S. Yan, and H. V. Poor, \textquotedblleft Ergodic capacity
analysis of remote radio head associations in cloud radio access
networks,\textquotedblright~\emph{IEEE Wireless Commun. Letters},
vol. 3, no. 4, pp. 365--368, Aug. 2014.

\bibitem{11}
M. Peng, Y. Li, J. Jiang, J. Li, and C. Wang, \textquotedblleft
Heterogeneous cloud radio access networks: A new perspective for
enhancing spectral and energy
efficiencies,\textquotedblright~\emph{IEEE Wireless Commun.}, vol.
21, no. 6, pp. 126--135, Dec. 2014

\bibitem{12}
M. Peng, Y. Li, Z. Zhao, and C. Wang, \textquotedblleft System
architecture and key technologies for 5G heterogeneous cloud radio
access networks,\textquotedblright~to appear in \emph{IEEE Network}.


\bibitem{13}
D. Lopez-Perez, I. Guvenc, G. de la Roche \emph{et al.},
\textquotedblleft Enhanced intercell interference coordination
challenges in heterogeneous networks,\textquotedblright~\emph{IEEE
Wireless Commun.}, vol. 18, no. 3, pp. 22--30, Mar. 2011.

\bibitem{13_1}
S. Zhang, S. Liew, and H. Wang \textquotedblleft Blind known
interference cancellation ,\textquotedblright~\emph{IEEE J. Sel.
Areas Commun.}, vol. 31, no. 8, pp. 1572--1582, Aug. 2013.

\bibitem{14}
W. Cheung, T. Quek, and M. Kountouris,
\textquotedblleft Throughput optimization, spectrum allocation, and access control in two-tier femtocell networks,\textquotedblright~\emph{IEEE J. Sel. Areas Commun.}, vol. 30, no. 3, pp. 561--574, Apr. 2012.

\bibitem{14_Lau}
A. Liu and V. Lau, \textquotedblleft Joint power and antenna
selection optimization in large cloud radio access
networks,\textquotedblright~ \emph{IEEE Trans. Signal Process.},
vol. 62, no. 5, pp. 1319-1328, Feb. 2014.

\bibitem{15}
J. Wang, M. Peng, S. Jin, and C. Zhao, \textquotedblleft A
generalized Nash equilibrium approach for robust cognitive radio
networks via generalized variational
inequalities,\textquotedblright~ \emph{IEEE Trans. Wireless
Commun.}, vol. 13, no. 7, pp. 3701--3714, Jul. 2014.

\bibitem{16}
S. Guruacharya, D. Niyato, I. Dong, and E. Hossain,
\textquotedblleft Hierarchical competition for downlink power
allocation in OFDMA femtocell networks,\textquotedblright~\emph{IEEE
Trans. Wireless Commun.}, vol. 12, no. 4, pp. 1543-1553, Apr. 2013.

\bibitem{17}
X. Kang, R. Zhang, and M. Motani, \textquotedblleft Price-based
resource allocation for spectrum-sharing femtocell networks: A
stackelberg game approach,\textquotedblright~\emph{IEEE J. Sel.
Areas Commun.}, vol. 30, no. 3, pp. 538--549, Apr. 2012.

\bibitem{18}
Q. Han, Y. Bo, X. Wang, K. Ma, and C. Chen, \textquotedblleft
Hierarchical-Game-Based uplink power control in femtocell
networks,\textquotedblright~\emph{IEEE Trans. Veh. Technol.}, vol.
63, no. 6, pp. 2819--2835, Jul. 2014.

\bibitem{19}
L. Duan, L. Gao, and J. Huang,
\textquotedblleft Cooperative spectrum sharing: a contract-based approach,\textquotedblright~\emph{IEEE Trans. Mobile Comput.}, vol. 13, no. 1, pp. 174--187, Jan. 2014.

\bibitem{20}
L. Gao, X. Wang, Y. Xu, and Q. Zhang, \textquotedblleft Spectrum
trading in cognitive radio networks: a contract-theoretic modeling
approach,\textquotedblright~\emph{IEEE J. Sel. Areas Commun.}, vol.
29, no. 4, pp. 843--855, Apr. 2011.

\bibitem{21}
Z. Hasan and V. Bhargava,
\textquotedblleft Relay selection for OFDM wireless systems under asymmetric information: a contract-theory based approach,\textquotedblright~\emph{IEEE Trans. Wireless Commun.}, vol. 12, no. 8, pp. 3824--3837, Aug. 2013.

\bibitem{22}
V. R. Cadambe and S. A. Jafar,
\textquotedblleft Interference alignment and degrees of freedom of the K-user interference channel,\textquotedblright~\emph{IEEE Trans. Inf. Theory}, vol. 54, no. 8, pp. 3425--3441, Aug. 2008.

\bibitem{23}
M. Biguesh and A. B. Gershman,\textquotedblleft Training-based MIMO channel estimation: a study of estimator tradeoffs and optimal training signals,\textquotedblright~
\emph{IEEE Trans. Signal Process.}, vol. 54, no. 3, pp. 884--893, Mar. 2006.

\bibitem{24}
N. Lee, O. Simeone, and J. Kang,
\textquotedblleft The effect of imperfect channel knowledge on a MIMO system with interference, \textquotedblright~\emph{IEEE Trans. Commun.}, vol. 60, no. 8, pp. 2221--2229, Aug. 2012.


\bibitem{26}
J. Choi and F. Adachi,
\textquotedblleft User selection criteria for multiuser systems with optimal and suboptimal LR based detectors,\textquotedblright~
\emph{IEEE Trans. Signal Process.}, vol. 58, no. 10, pp. 5463--5468, Oct. 2010.

\bibitem{27}
R. Rogolin, O. Y. Bursalioglu, H. Papadopoulos, G. Caire \emph{et
al.}, \textquotedblleft Scalable synchronization and reciprocity
calibration for distributed multiuser
MIMO,\textquotedblright~\emph{IEEE Trans. Wireless Commun.}, vol.
13, no. 4, Apr. 2014.

\bibitem{28}
M. Sichitiu and C. Veerarittiphan, \textquotedblleft Simple, accurate time synchronization for wireless sensor networks,\textquotedblright~in \emph{Proc. IEEE Wireless Communications and
Networking Conference (WCNC)}, New Orleans, LA, Mar. 2003, pp. 1266--1273.

\bibitem{29}
Y. Jing and S. ShahbazPanahi,
\textquotedblleft Max-min optimal joint power control and distributed beamforming for
two-way relay networks under per-node power constraints,\textquotedblright~
\emph{IEEE Trans. Signal Process.}, vol. 60, no. 12, pp. 6576--6589, Dec. 2012.

\bibitem{30}
B.Patrick and M. Dewatripont, \emph{Contract Theory},~MIT Press,
2005.

\bibitem{31}
J. Jose, A. Ashikhmin, T. L. Marzetta, and S. Vishwanath,
\textquotedblleft Pilot contamination and precoding in multi-cell
TDD systems, \textquotedblright~\emph{IEEE Trans. Wireless Commun.},
vol. 10, no. 8, pp. 2640--2650, Aug. 2011.

\end{thebibliography}
\end{document}